\documentclass[sigconf]{acmart}
\usepackage{amssymb}

\usepackage{subfigure}
\usepackage{algorithm}
\usepackage{algorithmic}
\usepackage{array} 
\graphicspath{{figures/}}
\newtheorem{theorem}{Theorem}
\newtheorem{lemma}{Lemma}
\usepackage{bm}

\AtBeginDocument{%
  \providecommand\BibTeX{{%
    \normalfont B\kern-0.5em{\scshape i\kern-0.25em b}\kern-0.8em\TeX}}}

\setcopyright{acmcopyright}
\copyrightyear{2020}
\acmYear{2020}
\acmDOI{10.1145/1122445.1122456}

\acmConference[CCS '20]{CCS '20: ACM Conference on Computer and Communications Security}{November 09-13, 2020}{Orlando, USA}
\acmBooktitle{CCS '20: ACM Conference on Computer and Communications Security,
  November 09-13, 2020, Orlando, USA}
\acmPrice{15.00}
\acmISBN{978-1-4503-XXXX-X/18/06}

\begin{document}

\title{Secure Computation Framework for Multiple Data Providers Against Malicious Adversaries}


\author{Zhili Chen, Xin Chen}
\affiliation{%
  \institution{School of Computer Science and Technology, Anhui University}
  \city{Hefei 230601}
  \country{China}}
\email{zlchen@ahu.edu.cn, 424049316@qq.com}








\begin{abstract}
  Due to the great development of secure multi-party computation, many practical secure computation schemes have been proposed. As an example, different secure auction mechanisms have been widely studied, which can protect bid privacy while satisfying various  economic properties. However, as far as we know, none of them solve the secure computation problems for multiple data providers (e.g., secure cloud resource auctions) in the malicious security model. In this paper, we use the techniques of cut-and-choose and garbled circuits to propose a general secure computation framework for multiple data providers against malicious adversaries. Specifically, our framework checks input consistency with the cut-and-choose paradigm, conducts maliciously secure computations by running two independent garbled circuits, and verifies the correctness of output by comparing two versions of outputs. Theoretical analysis shows that our framework is secure against a malicious computation party, or a subset of malicious data providers. Taking secure cloud resource auctions as an example, we implement our framework. Extensive experimental evaluations show that the performance of the proposed framework is acceptable in practice.


\end{abstract}




\maketitle

\section{Introduction}

The notion of secure multiparty computation (SMC) allows a number of parties to compute a function over their inputs jointly, while protects the input of each party from leaking to others. First introduced by Yao \cite{yao1986generate}, SMC has undergone a major development in the recent one or two decades. Particularly, in the two-party computation case, numerous practical secure computation schemes have been proposed, such as \cite{blanton2012secure,mood2014reuse,carter2014whitewash,carter2016secure} etc.

With the rapid development of cloud computing technology, cloud resource auction has attracted a lot of research attention. Amazon EC2 provides an effective auction mechanism called {\itshape Spot Instance}, which is used to allocate virtual machine (VM) instances. There are various other cloud resource auction mechanisms that satisfy different properties, such as truthfulness, social welfare maximization, etc \cite{wang2012cloud,zaman2013combinatorial,shi2014rsmoa,zhang2014dynamic}. However, the security of cloud resource auction is rarely taken into account. Without security guarantees, some sensitive information (e.g., bids, locations) may be leaked from the auction; the computation process may be deviated; and the auction output may be maliciously tampered. The security issues of cloud resource auction can become extremely important for companies or organizations, which may cause irreparable loss due to the unavailability of cloud resources.


To address the above problems, Chen {\itshape et al.} \cite{chen2016privacy} proposed a privacy-preserving cloud resource auction protocol using techniques of garbled circuits. The protocol guarantees that no more information about user bids will be disclosed except what can be inferred from the auction outcome, which means that the privacy of the auction is well protected. In a different auction scenario, Cheng {\itshape et al.} \cite{cheng2019towards} proposed a privacy-preserving double auction protocol for cloud markets based on secret sharing and garbled circuits. There is also a line of work on secure spectrum auctions \cite{chen2014ps,chen2015itsec,chen2016towards,chen2017secure}. These schemes preserved privacy for a variety of spectrum auction scenarios, by applying garbled circuits, secret sharing, homomorphic encryption, and others. However, current secure schemes for auctions, especially cloud resource auctions, only provides the security in the presence of semi-honest adversaries. Once there is a malicious party, these schemes are not sufficient to guarantee the security of auctions.


Motivated by the above observations, in this paper, we propose a general secure computation framework for multiple data providers against malicious adversaries. Our framework is originated from secure auctions, but is also suitable for specific scenarios of secure computations that meet the following conditions. First, the encrypted input data is provided by a number of data providers. Secondly, there are two non-colluding computation parties, who receive the encrypted input data and perform secure computations. Finally, the encrypted computation results are returned to data providers without the computation parties knowing anything about the plain results. Such kind of secure computations is also perfectly suitable for solving the hot problem of ``data island'' in the big data area, where multiple data owners want to analyze their data jointly and securely.

Our secure computation framework is depicted in Fig.~\ref{fig:framework}. To achieve security against malicious adversaries, two independent garbled circuits are run by two non-colluding computation parties $P_1$ and $P_2$, with the roles of parties swapped. Namely, for the first garbled circuit, $P_1$ is the garbler and $P_2$ is the evaluator, while for the second $P_2$ is the garbler and $P_1$ is the evaluator. For each garbled circuit, the data providers submit their inputs by committing to the encodings (i.e., pairs of $0$- and $1$- garbled values) and labels (i.e., $0$- or $1$- garbled values) of the inputs. They submit commitments to multiple independent copies of encodings and labels, so that the input consistency (that the input values for the two garbled circuits should be the same) can be checked through the cut-and-choose paradigm. Furthermore, the data providers, who submit inconsistent inputs, can be identified accurately with a cheating proof. In this way, for each garbled circuit computation, the garbler gets input encodings, and the evaluator gets input labels, without using oblivious transfers. The two independent garbled circuits are run simply using the standard techniques of garbled circuits, with exception that $P_1$ and $P_2$ do not decode the two garbled outputs. Instead, all the output encodings and labels are committed and published to all data providers, and the encodings and labels corresponding to its output are opened to each data provider. The data providers then decode the two versions of garbled outputs, and verify if both outputs are equal. The output verification results remain private from both $P_1$ and $P_2$, and can be proved to other data providers by revealing the correct openings from both $P_1$ and $P_2$.


\begin{figure}[ht]
  \centering
  \includegraphics[width=0.45\textwidth]{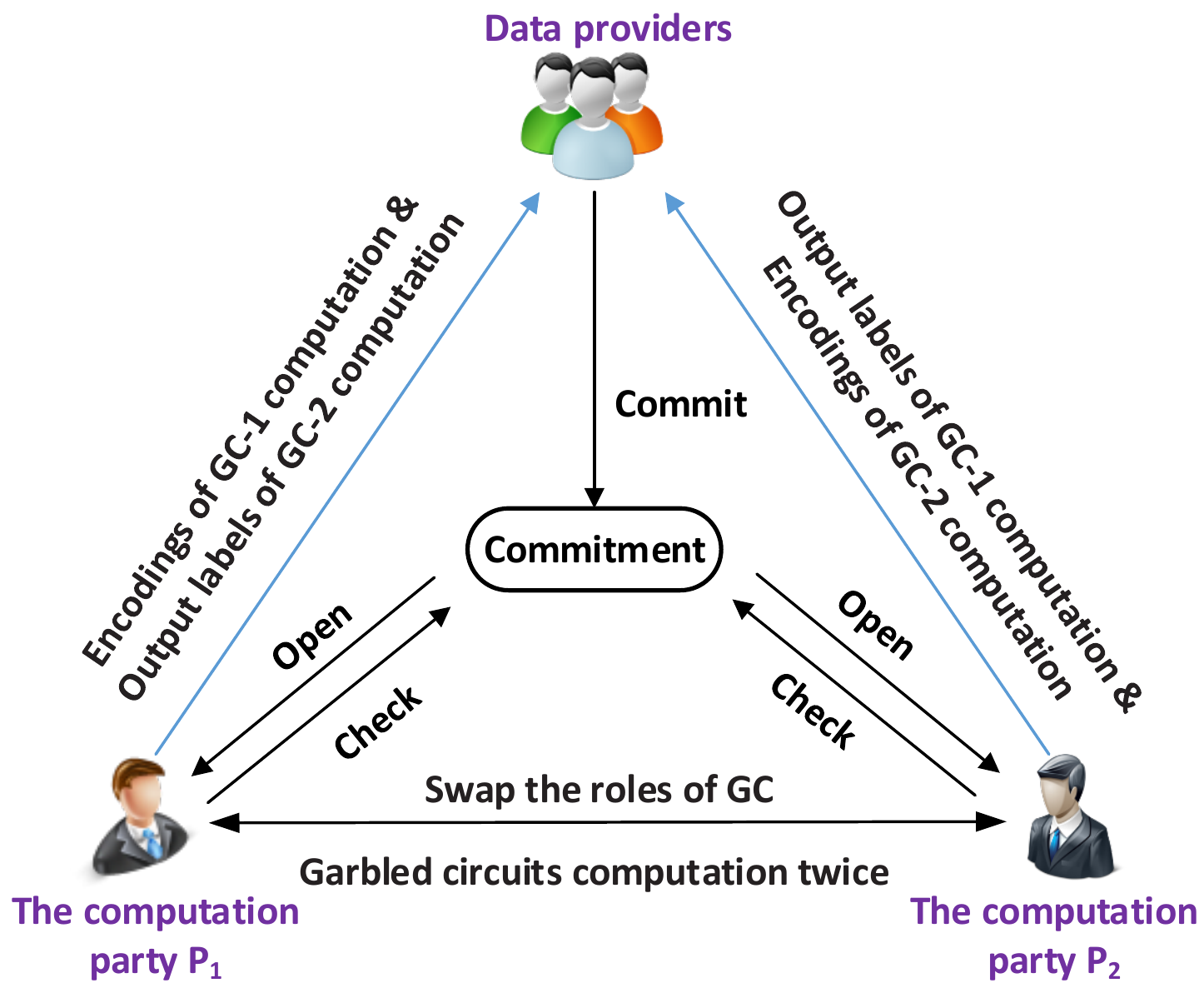}
  \caption{Secure computation framework for multiple data providers against malicious adversaries.}
  \label{fig:framework}
\end{figure}

The design of our secure computation framework for multiple data providers against malicious adversaries includes the following challenges. First of all, to achieve malicious security, it is necessary to run two independent garbled circuits, which leads to the problem of input consistency. It is challenging to guarantee the consistency of inputs to the two garbled circuits, while leaking no information about the inputs to both $P_1$ and $P_2$. Secondly, the logic of a garbled circuit may be tampered by the party who generates the circuit, how to prevent such attacks requires elaborate designs. Last but not least, in the presence of a malicious computation party or a set of malicious data providers, how to verify the outputs is still unknown. To address the above challenges, we should carefully design the input, computation, and output phases of the framework jointly, so that the entire protocol achieves security against malicious adversaries.

Our main contributions can be summarized as follows.
\begin{itemize}
  \item We propose a general secure computation framework for multiple data providers in the presence of malicious adversaries. In our context, this framework achieves malicious security by only running two independent garbled circuits, and eliminates the use of an oblivious transfer in garbled circuit computations.
  \item We design an input consistency check mechanism to ensure that all data providers submit the same input values for the computation of two independent garbled circuits. Furthermore, this mechanism can identify which data providers are cheating exactly by providing a cheating proof.

  \item We apply our framework to implement a secure cloud resource auction mechanism against malicious adversaries, and do extensive experiments to evaluate its performance. Experiments show that the performance of our framework is acceptable in practice.
\end{itemize}

The remainder of the paper is organized as follows. Section~\ref{sec:relatedwork} reviews related work briefly. Section~\ref{sec:preliminaries} introduces some technical preliminaries. We present the detailed design and security analysis of our framework in Section~\ref{sec:protocolframework}. In Section~\ref{sec:experiments}, we apply our framework to implement a secure cloud resource auction, and evaluate the performance in term of computation and communication overheads. In the end, we summarize the paper in Section~\ref{sec:conclusion}.

\section{Related Work}
\label{sec:relatedwork}

\subsection{Secure Multiparty Computation}
There is a large amount of literature on secure multiparty computation. However, since in this paper we mainly use the techniques of garbled circuits and cut-and-choose, we review only related work concerning these two techniques.

There are numerous fundamental studies on garbled circuits. The``point-and-permute'' method \cite{beaver1990round} makes computing a garbled gate need only to decode one ciphertext, which greatly improves the computation efficiency. The technique of ``Free-XOR'' \cite{kolesnikov2008improved} enables that the computation of garbled XOR gate only needs to XOR the corresponding input labels directly, which  improves both computation and communication efficiency. In \cite{kolesnikov2009improved}, the authors described some effective Boolean circuit for garbled circuit construction. These studies mainly focus on improving the performance of garbled circuits in the semi-honest setting, and they are compatible with our work.

Cut-and-choose is usually used as a generic technique for enabling security against malicious adversaries. In \cite{hazay2010efficient}, the authors showed how to obtain secure two-party computation in the presence of malicious adversaries by the cut-and-choose paradigm. Later, Huang {\itshape et al.} \cite{huang2013efficient} proposed a secure two-party computation protocol using a symmetric cut-and-choose to improve efficiency. This line of work also includes \cite{lindell2015efficient,lindell2012secure,lindell2016fast,kolesnikov2017duplo,katz2018optimizing}, and these studies have improved on the original cut-and-choose scheme in different methods. There are still schemes secure against malicious adversaries without using cut-and-choose. For example, Wang {\itshape et al.} \cite{wang2017authenticated} proposed an effective constant-round protocol framework for maliciously secure two-party computation, by constructing a single authenticated garbled circuit. Cut-and-choose garbled circuits are still the mainstream methods for achieving security against malicious adversaries due to efficiency. Different from these studies, our work focus on secure computation in a specific context such as auctions. By using the context, we are able to achieve security against adversaries with only two independent runs of garbled circuits, and thus greatly improve the efficiency.

There are also a few secure computation schemes that consider specific contexts and achieve security against malicious adversaries. The closest work is literature \cite{blanton2012secure}, in which the authors proposed a secure and efficient outsourcing scheme for sequence
comparisons against malicious adversaries, also with two independent garbled circuit runs. Indeed, we get some inspirations from this work. However, this work considered only a single data provider and assumed implicitly that the data provider is honest. This is a strong assumption, and reduces its applicability. Instead, our work considers multiple data providers who may be maliciously corrupted. The schemes \cite{carter2014whitewash,carter2016secure} securely outsource cut-and-choose garbled circuit computation to a cloud provider in the presence of malicious adversaries, greatly reducing the overhead of mobile devices who conduct the computations. Their computation models are quite different from ours in that there is only a mobile user involved in the secure computation. Paper \cite{mood2014reuse} proposed a secure function evaluation system that allows the reuse of encrypted values generated in a garbled circuit computation. This work is compatible with our work, and can be applied in our framework in the future.


\subsection{Secure Cloud Resource Auction}
In the past few years, a large number of truthful auction mechanisms for cloud resource allocation have been proposed. Wang {\itshape et al.} \cite{wang2012cloud} proposed an effective truthful auction mechanism, which used a greedy allocation and a critical value payment to enable fair competition for cloud resources. In \cite{zaman2013combinatorial2}, the authors proposed two allocation mechanisms based on combinatorial auction, CA-LP and CA-GREEDY, to allocate cloud virtual machine resources. The mechanisms based on combinatorial auction can significantly improve the efficiency of allocation and generate higher benefits to cloud providers. Papers \cite{zaman2011combinatorial,zaman2013combinatorial,nejad2015truthful,zhang2018truthful} have presented effective mechanisms for dynamic virtual machine resource allocation. Unfortunately, none of these studies considered the issue of security.

Recently, a few studies on secure cloud resource auction have been proposed. Chen {\itshape et al.} \cite{chen2016privacy} proposed a privacy-preserving cloud resource auction protocol, and used garbled circuits to protect the bid privacy in the process of the auction. Cheng {\itshape et al.} \cite{cheng2019towards} proposed PDAM, a privacy-preserving double auction mechanism for cloud markets based on secret sharing and garbled circuits. Xu {\itshape et al.} \cite{xu2019privacy} designed a privacy-preserving double auction mechanism based on homomorphic encryption and sorting networks. A related research topic that we need to mention is secure spectrum auction. Literature \cite{chen2014ps,chen2015itsec,chen2016towards,chen2017secure} proposed a number of secure spectrum auction mechanisms for a variety of auction scenarios, by applying garbled circuits, secret sharing, homomorphic encryption, and other techniques.
These secure auction mechanisms considered only security in the presence of semi-honest adversaries. Instead, our secure computation framework achieves security against malicious adversaries in this context for the first time.

There is also a line of work on auctions with differential privacy, such as \cite{zhu2014differentially,xu2017pads,chen2019differentially}. However, differential privacy aims to protect privacy from analyzing published output by adding noise to the original output, which is quite different from secure multiparty computation fundamentally.


\section{Technical Preliminaries}\label{sec:preliminaries}
In this section, we introduce some technical preliminaries for our framework.

\subsection{Garbled Circuits}
Garbled circuits are a cryptographic tool for secure multiparty computation \cite{yao1986generate,lindell2009proof}. During the process of garbled circuit computation, two parties hold their own private inputs respectively. Specifically, one holds its private input ${x_1}$, and the other holds its private input ${x_2}$. They want to compute a function $f$ over their inputs securely, without revealing any information about one party's private input to the other, except what can be inferred from the output. Then one party (called the garbler) prepares a garbled circuit to evaluate $f\left( {{x_1},{x_2}} \right)$. The garbler transmits the garbled circuit, garbled input ${x_1}$, and the output decoding table to the other party (called the evaluator). The evaluator obtains its garbled input ${x_2}$ obliviously by oblivious transfer, then computes the garbled output with both parties's garbled inputs, and then decodes the plain output (i.e., $f\left( {{x_1},{x_2}} \right)$) by the output decoding table.

\begin{figure}[ht]
\centering
\includegraphics[width=0.35\textwidth]{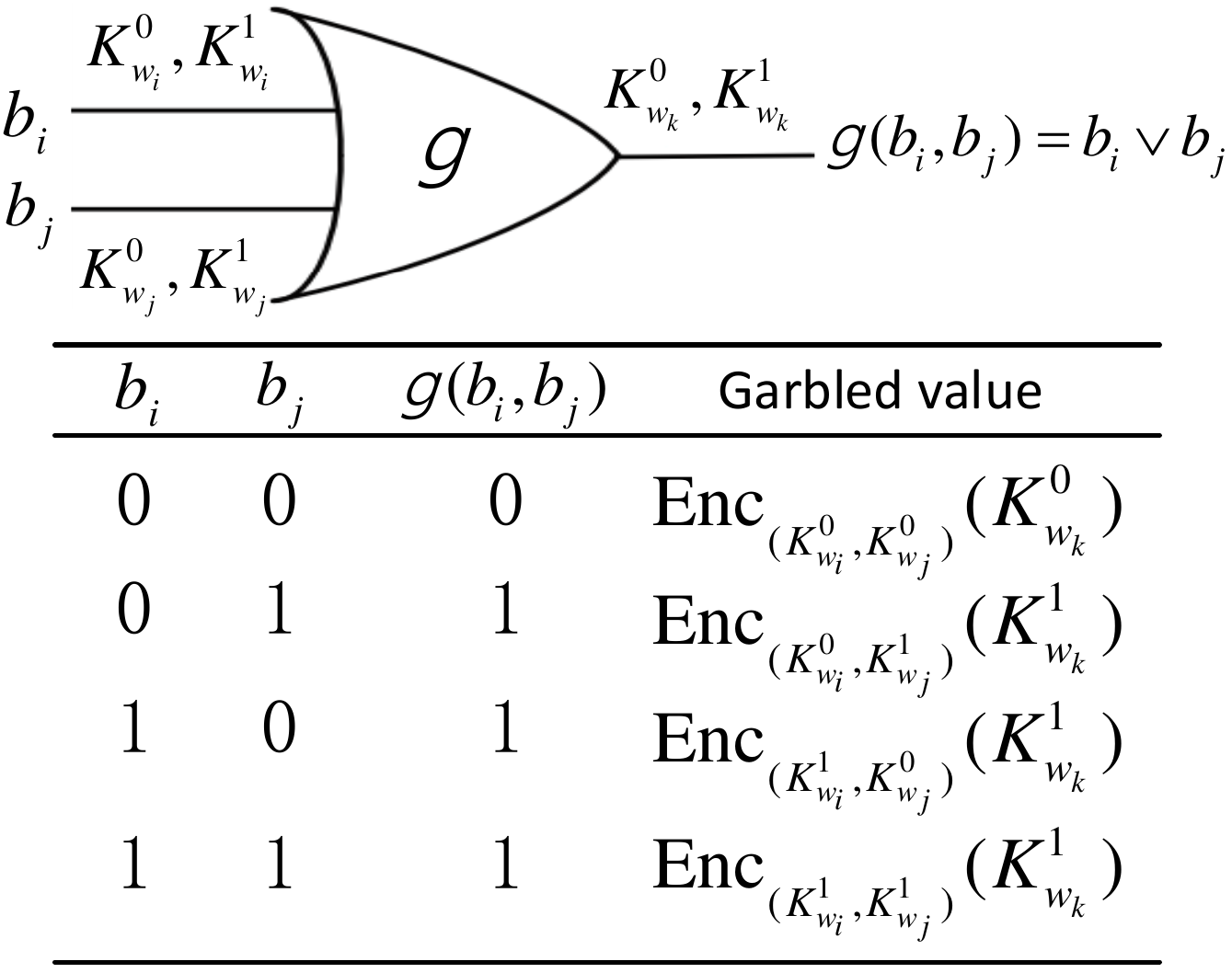}
\caption{Illustration of a garbled OR gate}
\label{fig:garbledcircuit}
\end{figure}

A garbled circuit can be constructed as follows. First, the garbler converts the computed function $f$ to a Boolean circuit containing only, without loss of generality, binary gates. Then, the garbler associates two random labels $K_{{w}}^0$ and $K_{{w}}^1$ to each wire ${w}$ of the circuit, encoding bit $0$ and bit $1$, respectively, and for each binary gate $g$ of the circuit with input wires ${w_i}$ and ${w_j}$, and output wire ${w_k}$, computes ciphertexts
\[{\rm{En}}{{\rm{c}}_{(K_{{w_i}}^{{b_i}},K_{{w_j}}^{{b_j}})}}(K_{{w_k}}^{g({b_i},{b_j})})\]
for $\ {b_i},{b_j} \in \{ 0,1\}$. A garbled gate is composed of above four ciphertexts. Fig.~\ref{fig:garbledcircuit} illustrates a garbled OR gate. Here, $Enc$ denotes the symmetric encryption algorithm used for garbled circuits. Finally, all the garbled binary gates constitute the garbled circuit of function $f$.

Given input labels $K_{{w_i}}^{{b_i}},K_{{w_j}}^{{b_j}}$, the evaluator can compute a garbled gate by recovering its corresponding output label $K_{{w_k}}^{g({b_i},{b_j})}$. Then, a garbled circuit which consists of a number of garbled gates can be computed by gate by gate.

In the process of a standard garbled circuit computation, the evaluator needs to run a 1-out-of-2 oblivious transfer protocol with the garbler \cite{rabin2005exchange} to get its input labels. In this paper, the data providers generate their input encodings and labels, and send input encodings to the garbler and input labels to the evaluator, thus eliminating the use of oblivious transfers.

\subsection{Cut-and-choose}
The technique of cut-and-choose based on garbled circuits \cite{lindell2015efficient} is used to construct a secure two-party computation protocol in the presence of malicious adversaries. It works as follows. The garbler prepares multiple copies of garbled circuits and sends the commitments of these circuits to the evaluator. Then, the evaluator selects randomly a set of circuits (called check circuits) and asks the garbler to open the commitments of them. After the garbler opens these commitments, the evaluator checks whether these circuits are correctly constructed. If all the check circuits are correct, the evaluator evaluates the circuits whose commitments have not been opened (called evaluation circuits), and takes most of the results of evaluation circuits as the final output. Otherwise, the evaluator will terminate the protocol. The method of cut-and-choose forces the garbler to construct a garbled circuit exactly. If a malicious garbler constructs the circuit by mistake, the evaluator will detect it with a high probability.

In our paper, we use the idea of cut-and-choose to design an input consistency check. This check ensures that data providers submit the same input for the two independent garbled circuits, with a proof indicating the data providers who is cheating.


\subsection{Problem Statement}
In our secure computation framework illustrated by Fig.~\ref{fig:framework}, there are $n$ data providers $(n \ge 1)$, and each data provider $u$ holds its private input ${x_u} \in \{0,1\}^l$, for $u \in [1..n]$. All the data providers want to compute a function over their inputs $f(x_1,x_2, \cdots , x_n)=\{y_u\}_{u \in [1..n]}$, where $y_u \in \{0,1\}^{l'}$ is data provider $u$'s corresponding output, and $l'$ is the bit length of each output. The data providers submit their inputs to two non-colluding computation parties $P_1$ and $P_2$, which cooperate to compute the function $f$ securely, and send each data provider its output.

The security goal is that all data providers want to compute the function $f$ correctly, and obtain their respective output without leaking anything about any data provider's input to others (including other data providers, $P_1$ and $P_2$). The threat model is that either $P_1$ or $P_2$, or a subset of data providers (in extreme case all the data providers) may be corrupted maliciously. We aim to design a protocol that securely computes the function $f$ in this security model.

\section{Our Framework}\label{sec:protocolframework}
In this section, we present our secure computation framework for multiple data providers, and analyze its security in the presence of malicious adversaries. Table~\ref{tab:notations} describes the notation we use in our protocol.

\begin{table}[tbp]\small
  \centering
  \caption{\small Notation for our framework}
  \label{tab:notations}
  \begin{tabular}{lm{0.35\textwidth}}
    \toprule
    Notation&Descriptions\\
    \midrule
    $f$ & The function computed.\\
    \midrule
    $x_u, y_u$ & Input and output of the data provider $u$.\\
    \midrule
    $n$ & The number of data providers.\\
    \midrule
    $s$ & The number of pairs of comment sets.\\
    \midrule
    $l, l'$ & Bit length of each input and output.\\
    \midrule
    $\rho$ & The challenge string for cut-and-choose.\\
    \midrule
    $W_{u,i,j}$ & First set of $j$-th pair of commitment sets on $u$'s $i$-th input wire.\\
    \midrule
    ${W'}_{u,i,j}$ & Second set of $j$-th pair of commitment sets on $u$'s $i$-th input wire.\\
    \midrule
    $K_{u,i,j}^{k,b}$ & Label assigned to $b$ in $j$-th pair of commitment sets on $u$'s $i$-th input wire for $k$-th garbled circuit.\\
    \midrule
    $K_{u,i}^{k,b}$ & Final label assigned to $b$ on $u$'s $i$-th input wire for $k$-th garbled circuit.\\
    \midrule
    $\overline{K}_{u,i}^{k,b}$ & Label assigned to $b$ on $u$'s $i$-th output wire for $k$-th garbled circuit.\\
    \midrule
    $H_{u,i}^{k,b}$ & Final hash value assigned to $b$ on $u$'s $i$-th input wire for $k$-th garbled circuit.\\
    \midrule
    $E_k$ & The output encodings of $k$-th garbled circuit.\\
    \midrule
    $O_k$ & The output labels of $k$-th garbled circuit.\\
    \midrule
    $E_{k,u}$ & The output encodings of $u$ of $k$-th garbled circuit.\\
    \midrule
    $O_{k,u}$ & The output labels of $u$ of $k$-th garbled circuit.\\
  \bottomrule
\end{tabular}
\end{table}

\subsection{Design Rationale}

In our framework, the data providers submit their inputs in proper form to two non-colluding computation parties $P_1$ and $P_2$, who run two independent garbled circuits of the same computation, with the roles of $P_1$ and $P_2$ swapped. The security against malicious adversaries is achieved by considering together the three phases, namely input consistency check, garbled circuit computation, and output verification.

\emph{Input consistency check.} The input consistency check uses the cut-and-choose paradigm. It ensures that each data provider provides labels representing the same values for the two independent garbled circuit computations. To prepare its input submission, each data provider $u$ first generates $s$ independent copies of encodings and labels corresponding to $u$'s input bits, where each copy contains encodings and labels for two independent garbled circuits. Here, by an encoding we mean a pair of $0$- and $1$- labels. Note also that the way of input eliminates the uses of oblivious transfer, and thus can improve the running efficiency. Next, each data provider $u$ commits to each copy of encodings and labels, and sends the commitments to both $P_1$ and $P_2$. Then, both parties randomly open some of the $s$ commitments of copies (i.e., check sets) to check if the commitments are well constructed, and each party partially opens the remaining commitments (i.e., evaluation sets) to get its corresponding encodings and labels, which are XORed together for two garbled circuits computations, and both parties check if these encodings and labels are consistent cooperatively. Through the consistency checks for both check and evaluation sets, our design ensures that the data provider $u$ can only cheat when it provides consistent labels in all check sets, and inconsistent labels in all evaluation sets, resulting in a failure probability $2^{-s+1}$.

\emph{Garbled circuit computation.} In this phase, two independent garbled circuits are computed without inputs and outputs learnt by the two computation parties. The first one is generated by $P_1$ and evaluated by $P_2$, while the second is generated by $P_2$ and evaluated by $P_1$. In the computation, an evaluator only learns input labels, while a garbler learns only input encodings. Both parties learn nothing about the plain inputs as long as they do not collude with each other. Similarly, when the output labels are obtained, an evaluator only learns output labels, while a garbler learns only output encodings. The two parties learn nothing about the plain output. Additionally, the standard garbled circuit techniques already offer protection against a malicious evaluator. Therefore, the run of two independent garbled circuits, with both parties learning neither inputs nor outputs, enables the malicious security, since in this case generating and evaluating an incorrect circuit poses no security risks on the data providers.

\emph{Output verification.} All the output labels and encodings are committed and published by both parties for verification. The commitments to labels and encodings corresponding to the output of each data provider $u$ are then opened to $u$, who decodes the labels with encodings to get two versions of plain outputs. Each data provider $u$ checks if the two versions of outputs are equal. If the check fails, $u$ provides a proof by revealing the openings for its output commitments. If there is no failure proof, each data provider $u$ accepts the output. Note that both computation parties do not know the result of output verification, and an attack by generating and evaluating an incorrect circuit cannot succeed.

\subsection{Design Details}
Protocol~\ref{protocol:framwork} describes the overall framework. Specifically, our protocol is divided into three phases: input consistency check, garbled circuit computation and output verification. The details of each phase are described as follows.

\floatname{algorithm}{Protocol}
\renewcommand{\algorithmicrequire}{\bfseries Input:}   
\renewcommand{\algorithmicensure}{\bfseries Output:}  

\begin{algorithm}[htb]
\caption{Secure computation framework for multiple data providers against malicious adversaries} \label{protocol:framwork}
\begin{algorithmic}[1]
\REQUIRE Each data provider $u$ hold its input $x_u$.
\ENSURE $P_1$ and $P_2$ get $\bot$, and each data provider $u$ gets $y_u$ if output verification succeeds, $\bot$ otherwise.

{\bfseries Phase 1: Input Consistency Check}\\
\underline{Each data provider $u$:}

\STATE Computes $s$ pairs of commitment sets $\left\{ {{W_{u,i,j}},{W_{u,i,j}'}} \right\}_{j = 1}^s$ by Eqs.~\eqref{equ:W} and \eqref{equ:Wprime}, for the $i$-th input wire.

\STATE Computes a position set $\{ com({b_{u,i,j}})\}_{j=1}^s$, ${b_{u,i,j}} \in \{ 0,1\} $, in term of $x_u$.

\STATE Sends all above commitments to parties ${P_1}$ and ${P_2}$.

\underline{Parties ${P_1}$ and ${P_2}$:}

For each input wire $i$ of each $u$, do Lines~\ref{line:ccc} to \ref{line:evaluate}.




\STATE Do \emph{commitment construction check}: randomly open and check some pairs of commitment sets.\label{line:ccc}

\STATE Do \emph{label consistency check}: check the consistency of final labels in evaluation sets, and give a consistency proof.\label{line:lcc}

\STATE Evaluate the final labels by Eq.~\eqref{equ:Kkq} for the two garbled circuits, respectively.\label{line:evaluate}


{\bfseries Phase 2: Garbled circuit computation}\\

\STATE $P_1$ generates garbled circuit $GC_1$, $P_2$ generates $GC_2$, and they send their respective circuits to each other.

\STATE $P_1$ evaluates $GC_2$ to get output labels $O_2$ while $P_2$ holds output encodings $E_2$, and $P_2$ evaluates $GC_1$ to get $O_1$, while $P_1$ holds $E_1$.


{\bfseries Phase 3: Output verification}\\

\STATE $P_1$ and $P_2$ commit to output encodings and labels by Eq.~\eqref{equ:com-e1o2} and \eqref{equ:com-e2o1}, and publish all commitments.

\STATE The commitments, $\{com(E_{k,u}), com(O_{k,u})\}_{k=1}^2$, related to each data provider $u$'s output is opened to $u$.

\STATE Each data provider $u$ computes two versions of outputs from $\{E_{k,u}, O_{k,u}\}_{k=1}^2$, and accepts the output if both versions of outputs are equal for all $u$.


\end{algorithmic}
\end{algorithm}

\subsubsection{\bfseries Phase 1: Input Consistency Check}
To keep the consistency of the inputs, we design a mechanism which combines a commitment scheme and the idea of cut-and-choose at this phase. The data providers submit the labels of their inputs twice, then the computation parties ${P_1}$ and ${P_2}$ verify the consistency of two submissions. This phase consists of the following three steps.

{\bfseries Step 1: Input commitment.}

Suppose the bit length of all inputs is $l$. Each data provider $u$ generates $s$ pairs of commitment sets for each of its input wires. Specifically, for the $i$-th ($i \in \{ 1,l\} $) bit of its input, data provider $u$ generates $s$ pairs of commitment sets $\left\{ {{W_{u,i,j}},{W_{u,i,j}'}} \right\}_{j = 1}^s$, and $W_{u,i,j}$, $W'_{u,i,j}$ are defined as follows.
\begin{equation}\label{equ:W}
\begin{aligned}
{W_{u,i,j}} = \{ &com(K_{u,i,j}^{1,0}\| K_{u,i,j}^{1,1}\| K_{u,i,j}^{2,b}),\\
&com(K_{u,i,j}^{2,0}\| K_{u,i,j}^{2,1}\| K_{u,i,j}^{1,b})\}
\end{aligned}
\end{equation}
\begin{equation}\label{equ:Wprime}
\begin{aligned}
{W_{u,i,j}'} = \{ &com(K_{u,i,j}^{1,0}\| K_{u,i,j}^{1,1}\| K_{u,i,j}^{2,1 - b}),\\
&com(K_{u,i,j}^{2,0}\| K_{u,i,j}^{2,1}\| K_{u,i,j}^{1,1 - b})\}
\end{aligned}
\end{equation}
Where $com$ refers to a perfectly binding commitment scheme \cite{goldreich2007foundations}, and symbol $\parallel$ denotes connection of strings, and $b \in \{ 0,1\} $ is chosen randomly and independently for each $u$, $i$, $j$.

The commitment set $W_{u,i,j}$ (resp. $W'_{u,i,j}$) contains two commitments corresponding to the case that data provider $u$'s $i$-th input wire is assigned to bit value $b$ (resp. $1-b$). In the first commitment, ($K_{u,i,j}^{1,0}$, $K_{u,i,j}^{1,1}$) represents the $j$-th encoding of $u$'s $i$-th input wire in the first garbled circuit, and $K_{u,i,j}^{2,b}$ (resp. $K_{u,i,j}^{2,1-b}$) represents the $j$-th label of $u$'s $i$-th input wire in the second garbled circuit. Similarly, in the second commitment, ($K_{u,i,j}^{2,0}$, $K_{u,i,j}^{2,1}$) represents the $j$-th encoding of $u$'s $i$-th input wire in the second garbled circuit, and $K_{u,i,j}^{1,b}$ (resp. $K_{u,i,j}^{1,1-b}$) represents the $j$-th label of $u$'s $i$-th input wire in the first garbled circuit. The two commitments will be opened to parties $P_1$ and $P_2$, respectively, enabling them to compute two independent garbled circuits with party roles exchanged. The two commitments guarantee that the $j$-th labels of $u$'s $i$-th input wire sent to both parties are consistent (i.e., representing the same bit value, either $b$ or $1-b$), and the randomness of $b$ ensures that the bit value is not leaked even when both $W_{u,i,j}$ and $W'_{u,i,j}$ are opened for check, or either of them is opened for evaluation.

Moreover, each data provider $u$ commits to a position indicating which commitment set in each pair it chooses as input commitment set according to its input bit. These input commitment sets will be opened for evaluation later. The commitments to positions constitute a position set, denoted by $\{ com({b_{u,i,j}})\}_{j=1}^s$, where ${b_{u,i,j}} \in \{ 0,1\}$, and the $j$-th input commitment set of $u$'s $i$-th input wire is indicated to be ${W_{u,i,j}}$ if ${b_{u,i,j}}=0$, and ${W'_{u,i,j}}$ otherwise. For example, if the position set of $u$'s first input wire is $\{ com(0),com(1), \cdots ,com(0)\}$, it means that the corresponding input commitment sets are ${W_{u,1,1}},{W_{u,1,2}'}, \ldots ,{W_{u,1,s}}$, respectively. Fig.~\ref{fig:commitSets} illustrates above constructions.


\begin{figure}[ht]
\centering
\includegraphics[width=0.48\textwidth]{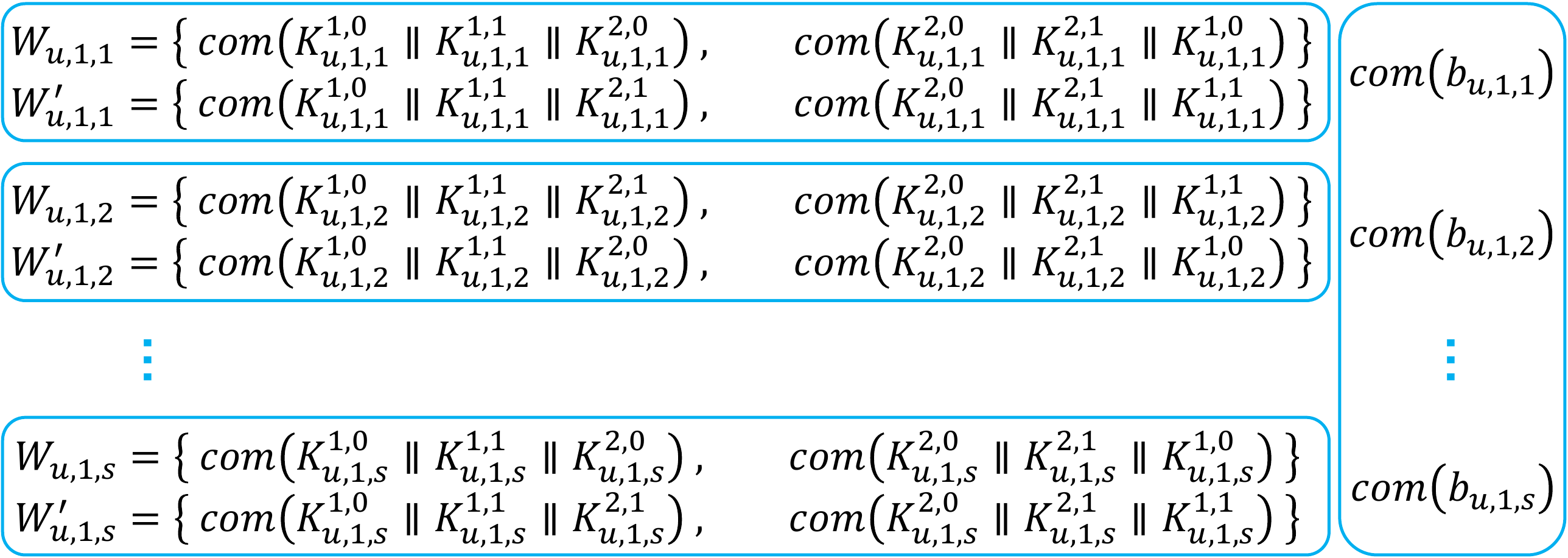}
\caption{The commitment sets corresponding to the first input wire of the data provider $u$}
\label{fig:commitSets}
\end{figure}



Each data provider $u$ sends all the commitment sets constructed above on each of its input wires to both computation parties ${P_1}$ and ${P_2}$. If the input length is $l$, there are $ls$ pairs of commitment sets and a total of $ls \times (4 + 1) = O(ls)$ commitments to send.

{\bfseries Step 2: Consistency check.}

In this step, we first check if the input commitments are properly constructed by opening and checking some of the commitments, and then check if the remaining commitments satisfy label consistency for garbled circuit evaluation.

\emph{Commitment construction check.} To determine which pair of commitment sets are used for opening and checking, a challenge string is prepared using a coin-tossing protocol as follows for each input wire $i$ of each data provider $u$: 1) Party ${P_1}$ selects a random string ${\rho _1} \in {\{ 0,1\} ^s}$; 2) Party ${P_2}$ selects a random string ${\rho _2} \in {\{ 0,1\} ^s}$; 3) Both parties exchange random strings, and compute the final challenge string by $\rho  = {\rho _1} \oplus {\rho _2}$ . If the $j$-th bit of $\rho $ is 1, the pair of commitment sets $\{ {W_{u,i,j}},{W_{u,i,j}'}\} $ will be opened and checked, and we call opened pairs of commitment sets check sets. Similarly, if the $j$-th bit of $\rho $ is 0, the pair of commitment sets $\{ {W_{u,i,j}},{W_{u,i,j}'}\} $ will be used to evaluate the garbled input bit, and we call these pairs evaluation sets.

Each data provider $u$ opens the commitments in all check sets of each input wire $i$ to both ${P_1}$ and ${P_2}$ for checking. It also opens the commitments, corresponding to evaluation sets, in the position set of each input wire $i$ to both parties. For any input wire $i$, any $j$-th pair of commitment sets that is check set, both parties check if the two commitments in $W_{u,i,j}$ (resp. $W'_{u,i,j}$) are well constructed. Specifically, suppose that the two commitments of $W_{u,i,j}$ or $W'_{u,i,j}$ are opened as two triples $(K_{1,1}, K_{1,2}, K_{1,3})$ and $(K_{2,1}, K_{2,2}, K_{2,3})$, respectively. Parties $P_1$ and $P_2$ check if the combination of the first two values in each triple, namely, $(K_{1,1}, K_{1,2})$ and $(K_{2,1}, K_{2,2})$, form encodings (i.e., $K_{1,1} \ne K_{1,2}$ and $K_{2,1} \ne K_{2,2}$), and if the the third value of one triple is the label representing the same bit value of the encoding of the other triple (i.e., $K_{1,1}=K_{2,3}$ and $K_{2,1}=K_{1,3}$, or $K_{1,2}=K_{2,3}$ and $K_{2,2}=K_{1,3}$). If any of the checks fails, both parties output ``Bad Input'' and abort the protocol.

Fig.~\ref{fig:commitSetsopen} demonstrates the opening and checking process of $u$'s first input wire. First, the challenge string generated determines which pairs of commitment sets are check sets (in green solid circles) and which are evaluation sets (without green solid circles). Second, all commitment sets in check sets are opened, and in the position set, all commitments responding to evaluation sets are opened such that the input commitment sets (in red hollow circles) are determined. Finally, both $P_1$ and $P_2$ check if all the opened check sets are well constructed.



\begin{figure}[ht]
\centering
\includegraphics[width=0.48\textwidth]{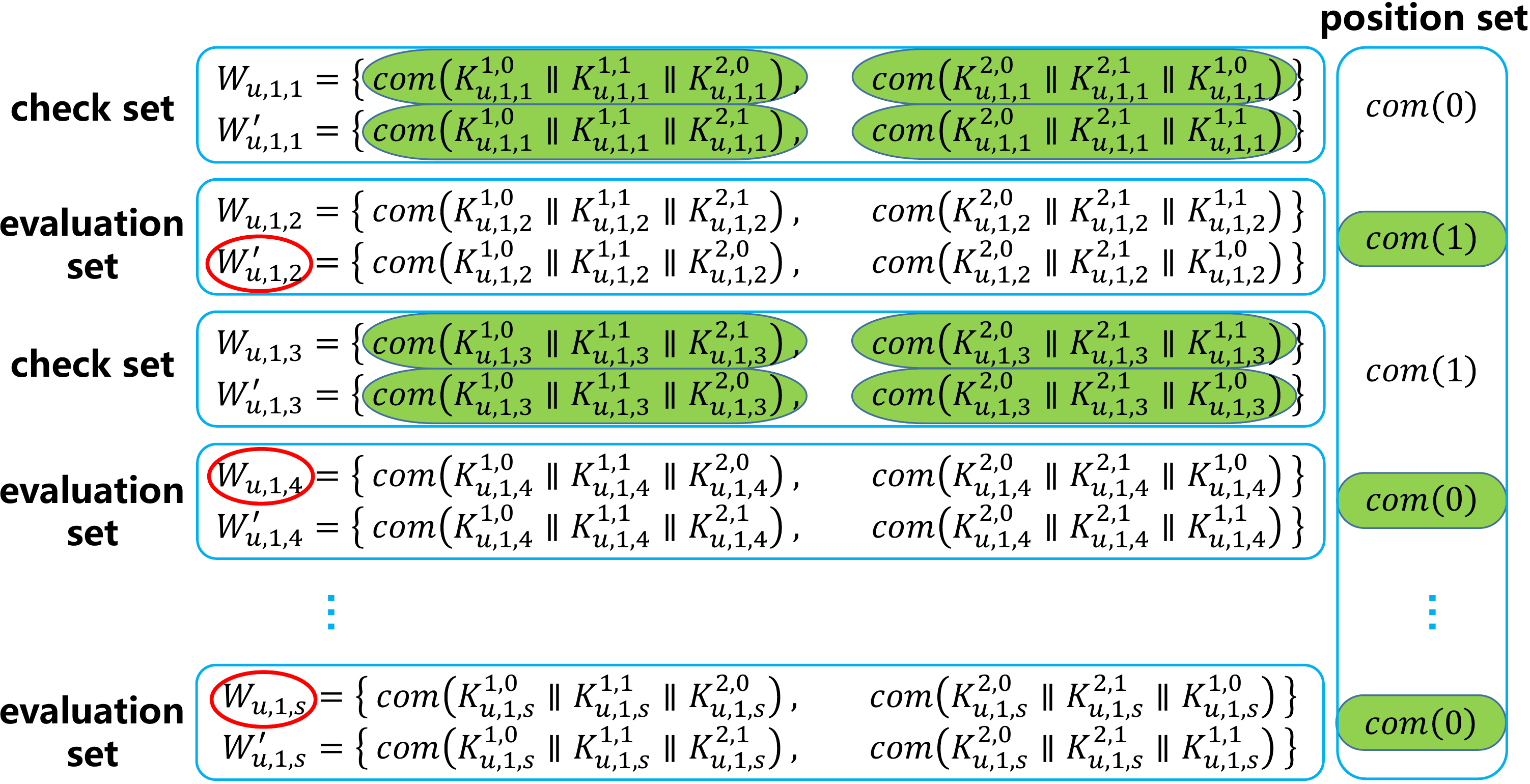}
\caption{\small An example of the commitment construction check for data provider $u$'s first input wire. The challenge string is $\rho =\left \langle1010 \cdots 0 \right \rangle $. The values of $b$'s are $ 0,1,1,0,\cdots,0$, and the position indicators committed are also $0,1,1,0,\cdots,0$, which implies the bit value of the input wire is $0$.}
\label{fig:commitSetsopen}
\end{figure}

\emph{Label consistency check.} The commitment construction check does not suffice for input consistency check. It is possible that the the commitment sets in the evaluation sets are not well constructed, or that a data provider does not commit to a correct position set. Thus, we design a label consistency check for labels committed by the input commitment sets.

Suppose the index set of evaluation sets for data provider $u$'s input wire $i$ is $\mathcal{J} \subset [1..s]$. Denote the input commitment set by $W^*_{u,i,j}$ , and the corresponding triples committed by $\{(K_j^{1,1}, K_j^{1,2}, K_j^{1,3}),$ $(K_j^{2,1}, K_j^{2,2}, K_j^{2,3})\}$, for $j \in \mathcal{J}$. Let $H(.)$ be a collision-resistant hash function. To start evaluating, the first commitments of $\{W^*_{u,i,j}\}_{j \in \mathcal{J}}$ are opened to party $P_1$, and the second ones are opened to party $P_2$. As a result, party $P_1$ gets triples $(K_j^{1,1}, K_j^{1,2}, K_j^{1,3})_{j \in \mathcal{J}}$, and $P_2$ gets triples $(K_j^{2,1}, K_j^{2,2}, K_j^{2,3})_{j \in \mathcal{J}}$, respectively. Both $P_1$ and $P_2$ check the consistency of these triples as follows.
\begin{itemize}
\item[1)] For $q \in \{1,2,3\}$, $P_1$ computes
\begin{equation}
H^{1,q} = \oplus_{j \in \mathcal{J}} H(K_j^{1,q})
\end{equation}
\begin{equation}
C^{1,q} = com(\|_{j \in \mathcal{J}} H(K_j^{1,q}))
\end{equation}

\item[2)] For $q \in \{1,2,3\}$, $P_2$ computes
\begin{equation}
H^{2,q} = \oplus_{j \in \mathcal{J}} H(K_j^{2,q})
\end{equation}
\begin{equation}
C^{2,q} = com(\|_{j \in \mathcal{J}} H(K_j^{2,q}))
\end{equation}

\item[3)] $P_1$ randomly reorders tuple $(H^{1,1}, H^{1,2})$ to get a new tuple $(\widetilde{H}^{1,1}, \widetilde{H}^{1,2})$ and the corresponding commitment tuple $(\widetilde{C}^{1,1}, \widetilde{C}^{1,2})$, which are sent to $P_2$.

\item[4)] $P_2$ randomly reorders tuple $(H^{2,1}, H^{2,2})$ to get a new tuple $(\widetilde{H}^{2,1}, \widetilde{H}^{2,2})$ and the corresponding commitment tuple $(\widetilde{C}^{2,1}, \widetilde{C}^{2,2})$, which are sent to $P_1$.

\item[5)] $P_1$ checks if $H^{1,3} \in \{\widetilde{H}^{2,1}, \widetilde{H}^{2,2}\}$. If the check fails, $P_1$ outputs a message ``Bad Input'', and a proof consisting of the data provider $u$, the input wire $i$, triples $(\widetilde{H}^{2,1}, \widetilde{H}^{2,2}, H^{1,3})$ and $(\widetilde{C}^{2,1}, \widetilde{C}^{2,2}, C^{1,3})$, and then abort the protocol.
\item[6)] Similarly, $P_2$ checks if $H^{2,3} \in \{\widetilde{H}^{1,1}, \widetilde{H}^{1,2}\}$. If the check fails, $P_2$ acts in the same way except outputting a proof with triples $(\widetilde{H}^{1,1}, \widetilde{H}^{1,2}, H^{2,3})$ and $(\widetilde{C}^{1,1}, \widetilde{C}^{1,2}, C^{2,3})$.
\end{itemize}

If the above label consistency check fails, all the data providers are aware of which input wire of which data provider does not satisfy the consistency. They can also verify this result by examining the corresponding proof. Suppose a data provider $u_0$ received a proof indicating $u$'s input wire $i$ fails to satisfy consistency, with triples $(\widetilde{H}^{1,1}, \widetilde{H}^{1,2}, H^{2,3})$ and $(\widetilde{C}^{1,1}, \widetilde{C}^{1,2}, C^{2,3})$. The result can be verified as follows.
\begin{itemize}
\item[1)] $u_0$ requests $P_1$ to open commitments $\widetilde{C}^{1,1}$, $\widetilde{C}^{1,2}$, and obtains $\{H(\widetilde{K}_j^{1,1}),H(\widetilde{K}_j^{1,2})\}_{j \in \mathcal{J}}$; requests $P_2$ to open commitment $C^{2,3}$, and gets $\{H(K_j^{2,3})\}_{j \in \mathcal{J}}$.
\item[2)] $u_0$ checks if the following three equations hold at the same time.
\begin{equation}
\left\{
\begin{aligned}
\widetilde{H}^{1,1} =& \oplus_{j \in \mathcal{J}} H(\widetilde{K}_j^{1,1})\\
\widetilde{H}^{1,2} =& \oplus_{j \in \mathcal{J}} H(\widetilde{K}_j^{1,2})\\
H^{2,3} =& \oplus_{j \in \mathcal{J}} H(K_j^{2,3})
\end{aligned}
\right.
\end{equation}
If the check fails, $u_0$ concludes that the parties are cheating.
\item[3)] $u_0$ checks if $H(K_j^{2,3}) = H(\widetilde{K}_j^{1,1})$ for all $j \in \mathcal{J}$, or $H(K_j^{2,3}) = H(\widetilde{K}_j^{1,2})$ for all $j \in \mathcal{J}$. If this check fails, the label inconsistency verified.
\end{itemize}



Fig.~\ref{fig:check} illustrates how to check label consistency for data provider $u$'s first input wire in the example of Fig.~\ref{fig:commitSetsopen}. All the opened labels in the same position of input commitment sets of the wire input are hashed and XORed to get the final hash values, $H_{u,1}^{1,0}$, $H_{u,1}^{1,1}$, $H_{u,1}^{2,0}$ held by $P_1$, and $H_{u,1}^{2,0}$, $H_{u,1}^{2,1}$, $H_{u,1}^{1,0}$ held by $P_2$. $P_1$ randomly reorders pair ($H_{u,1}^{1,0}$, $H_{u,1}^{1,1}$) to get ($\widetilde{H}_{u,1}^{1,0}$, $\widetilde{H}_{u,1}^{1,1}$), sends the latter to $P_2$, who checks if $H_{u,1}^{1,0}$ is contained in $\{\widetilde{H}_{u,1}^{1,0}, \widetilde{H}_{u,1}^{1,1}\}$. Symmetrically, $P_2$ randomly reorders pair ($H_{u,1}^{2,0}$, $H_{u,1}^{2,1}$) to get ($\widetilde{H}_{u,1}^{2,0}$, $\widetilde{H}_{u,1}^{2,1}$), sends the latter to $P_1$, who checks if $H_{u,1}^{2,0}$ is contained in $\{\widetilde{H}_{u,1}^{2,0}, \widetilde{H}_{u,1}^{2,1}\}$. If both checks succeed, then the label consistency check succeeds.

\begin{figure}[ht]
\centering
\includegraphics[width=0.48\textwidth]{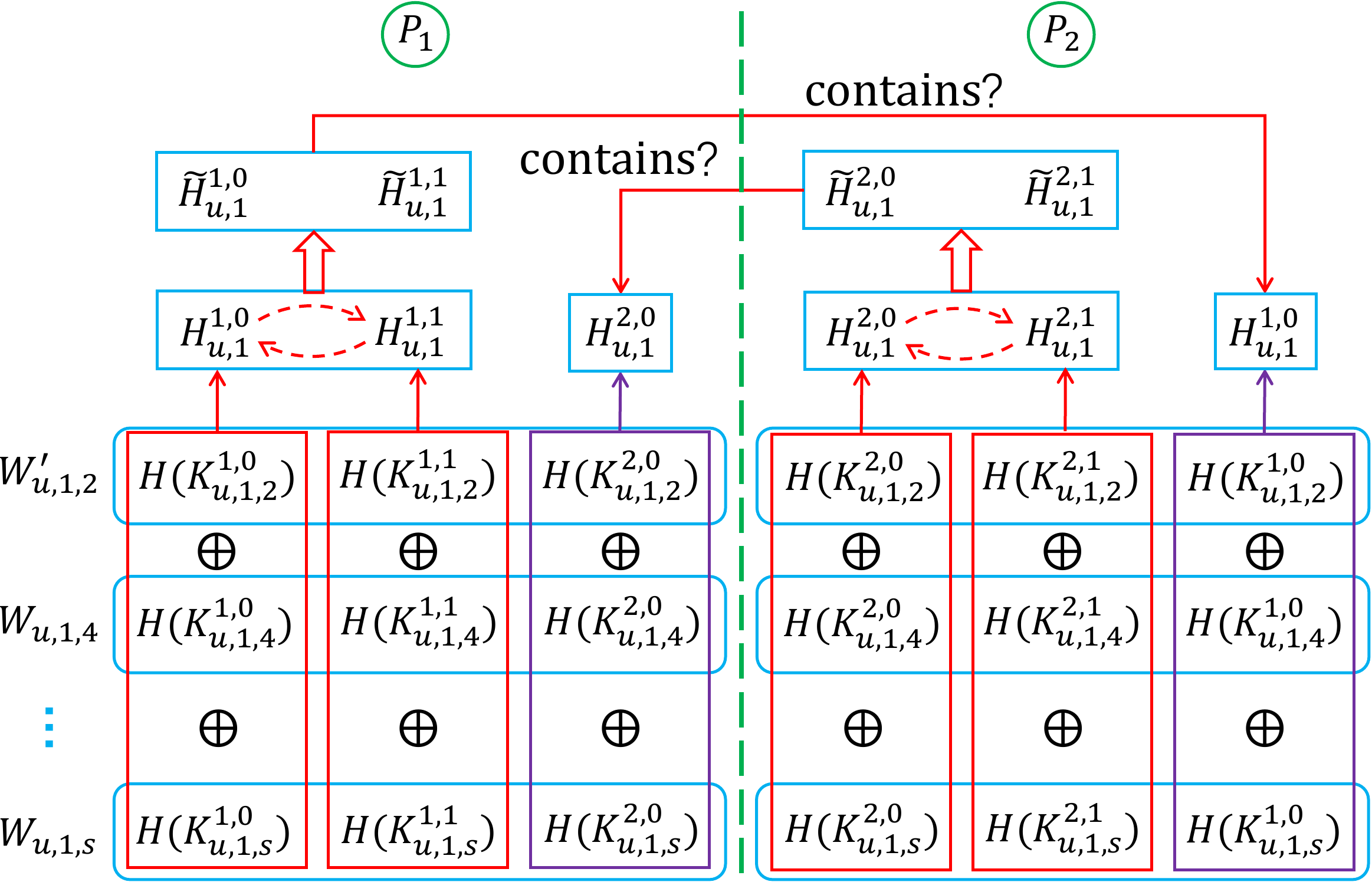}
\caption{The label consistency check for data provider $u$'s first input wire. }
\label{fig:check}
\end{figure}



{\bfseries Step 3: Label evaluation.}

Once the consistency check succeeds, parties $P_1$ and $P_2$ proceed to compute the final labels for each input wire of each data provider. Just use the same denotation as that of label consistency check. We take as an example data provider $u$'s input wire $i$, whose index set of evaluation sets is $\mathcal{J} \subset [1..s]$. The input commitment set $W^*_{u,i,j}$ commits to triples $\{(K_j^{1,1}, K_j^{1,2}, K_j^{1,3}),$ $(K_j^{2,1}, K_j^{2,2}, K_j^{2,3})\}$, for all $j \in \mathcal{J}$. Then, party $P_k$, for $k \in \{1,2\}$, computes its the final labels as follows.
\begin{equation}\label{equ:Kkq}
K^{k,q} = \oplus_{j \in \mathcal{J}} K_j^{k,q}
\end{equation}
for $q \in \{1,2,3\}$. Finally, concerning data provider $u$'s input wire $i$, party $P_1$ uses $(K^{1,1}, K^{1,2})$ as the encoding for the first garbled circuit, and uses $K^{1,3}$ as the label for the second garbled circuit; while party $P_2$ uses $(K^{2,1}, K^{2,2})$ as the encoding for the second garbled circuit, and uses $K^{2,3}$ as the label for the first garbled circuit.

Fig.~\ref{fig:commitSetsopen2} shows how to compute the final labels for data provider $u$'s first input wire in the example of Fig.~\ref{fig:commitSetsopen}. All the opened labels in the same position of the input wire's input commitment sets are XORed to get final labels. In the end, $P_1$ holds ($K_{u,1}^{1,0}$, $K_{u,1}^{1,1}$) as the encoding for the first garbled circuit, and $K_{u,1}^{2,0}$ as the label for the second garbled circuit. Symmetrically, $P_2$ holds ($K_{u,1}^{2,0}$,$K_{u,1}^{2,1}$) as the encoding for the second garbled circuit, and $K_{u,1}^{1,0}$ as the label for the first garbled circuit.

\begin{figure}[ht]
\centering
\includegraphics[width=0.48\textwidth]{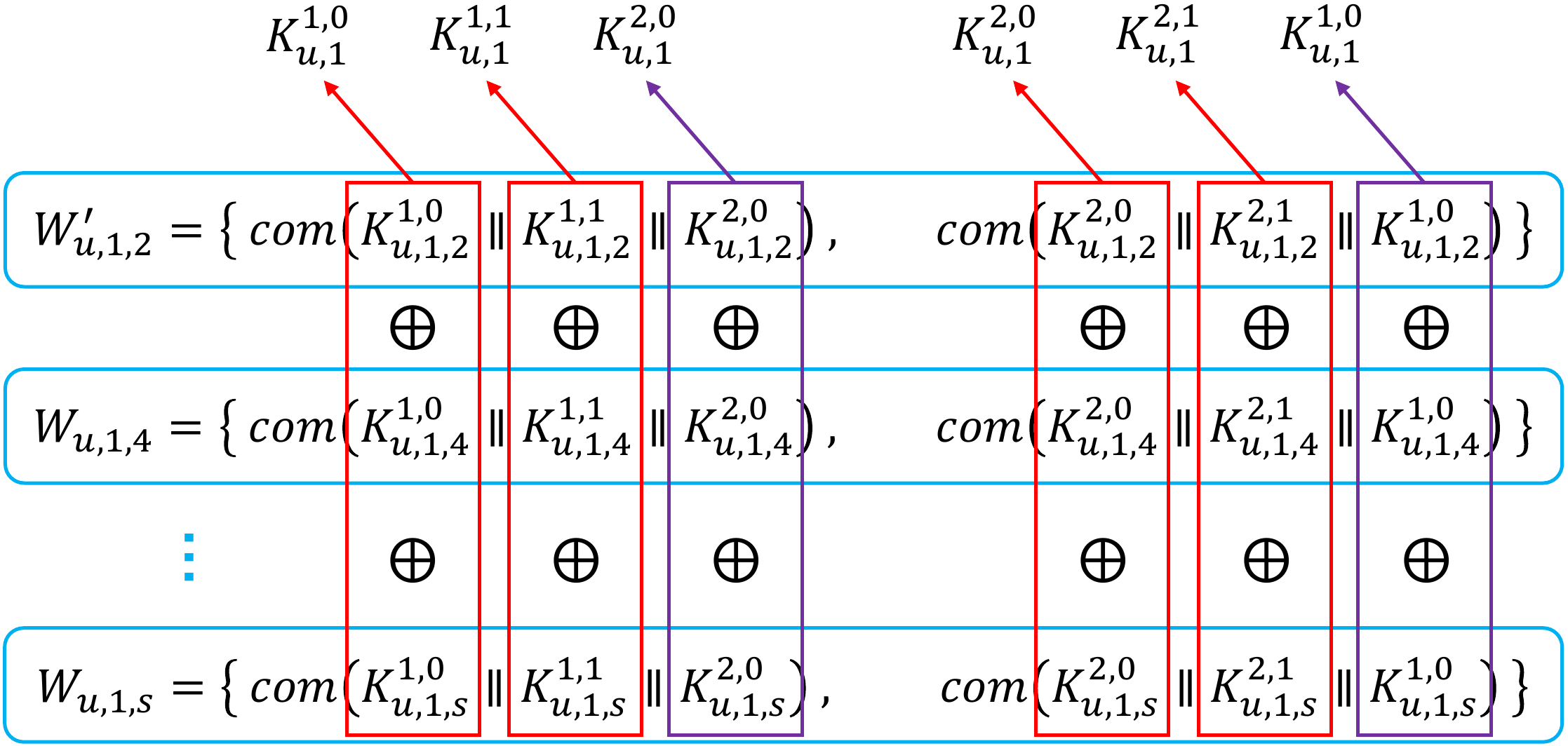}
\caption{Computation of final labels for data provider $u$'s first input wire. }
\label{fig:commitSetsopen2}
\end{figure}

\subsubsection{\bfseries Phase 2: Garbled Circuit Computation}
In this phase, the protocol performs two independent garbled circuit computations to prevent a malicious garbler that attacks by generating a falsified garbled circuit. The protocol proceeds as follows.


\begin{itemize}
\item[1)] ${P_1}$ uses $\{K_{u,i}^{1,0}, K_{u,i}^{1,1}\}_{u \in [1..n], i \in [1..l]}$ as the encodings of all data providers' input bits to generate the first garbled circuit $GC_1$, where $n$ is the number of data providers, and $l$ is the bit length of inputs. Denote the output encodings of $GC_1$ by $$E_1 = \{\overline{K}_{u,i}^{1,0}, \overline{K}_{u,i}^{1,1}\}_{u \in [1..n], i \in [1..l']}$$ where $l'$ is the bit length of outputs. $P_1$ sends $GC_1$ to $P_2$, while holds $E_1$ by itself.
\item[2)] ${P_2}$ uses $\{K_{u,i}^{2,0}, K_{u,i}^{2,1}\}_{u \in [1..n], i \in [1..l]}$ as the encodings of all input bits to generate the second garbled circuit $GC_2$. Denote the output encodings of $GC_2$ by $$E_2 = \{\overline{K}_{u,i}^{2,0}, \overline{K}_{u,i}^{2,1}\}_{u \in [1..n], i \in [1..l']}$$ $P_2$ sends $GC_2$ to $P_1$, while holds $E_2$ by itself.
\item[3)] After receiving $GC_2$, ${P_1}$ uses $\{K_{u,i}^{2,x_{u,i}}\}_{u \in [1..n], i \in [1..l]}$ as input labels to compute $GC_2$, obtaining the output labels $$O_2=\{\overline{K}_{u,i}^{2,y_{u,i}}\}_{u \in [1..n], i \in [1..l']}$$
\item[4)] After receiving $GC_1$, ${P_2}$ uses $\{K_{u,i}^{1,x_{u,i}}\}_{u \in [1..n], i \in [1..l]}$ as input labels to compute $GC_1$, obtaining the output labels $$O_1=\{\overline{K}_{u,i}^{2,y_{u,i}}\}_{u \in [1..n], i \in [1..l']}$$
\end{itemize}

At the end of this phase, $P_1$ holds $(E_1, O_2)$, and $P_2$ holds $(E_2, O_1)$, respectively.




\subsubsection{\bfseries Phase 3: Output Verification}
In this phase, both parties $P_1$ and $P_2$ commit to all the labels and encodings corresponding to the outputs of the two garbled circuits. Then, they open to each data provider $u$ the commitments to the labels and encodings corresponding to its output. Each data provider $u$ can verify if the outputs of the two garbled circuits are consistent, and all the data providers determine cooperatively whether to adapt the computed result. The protocol proceeds as follows.
\begin{itemize}
\item[1)] $P_1$ computes the commitments to the output encodings of the first garbled circuit and the output labels of the second garbled circuit.
\begin{equation}\label{equ:com-e1o2}
\begin{aligned}
com(E_1) =& \{com(E_{1,u})\}_{u \in [1..n]}\\
=& \{com(\|_{i=1}^{l'}(\overline{K}_{u,i}^{1,0} \| \overline{K}_{u,i}^{1,1}))\}_{u \in [1..n]}\\
com(O_2) =& \{com(O_{2,u})\}_{u \in [1..n]}\\
 =& \{com(\|_{i=1}^{l'}\overline{K}_{u,i}^{2,y_{u,i}})\}_{u \in [1..n]}
\end{aligned}
\end{equation}

\item[2)] Symmetrically, $P_2$ computes the commitments to the output encodings of the second garbled circuit and the output labels of the first garbled circuit.
\begin{equation}\label{equ:com-e2o1}
\begin{aligned}
com(E_2) =& \{com(E_{2,u})\}_{u \in [1..n]}\\
=& \{com(\|_{i=1}^{l'}(\overline{K}_{u,i}^{2,0} \| \overline{K}_{u,i}^{2,1}))\}_{u \in [1..n]}\\
com(O_1) =& \{com(O_{1,u})\}_{u \in [1..n]}\\
 =& \{com(\|_{i=1}^{l'}\overline{K}_{u,i}^{1,y_{u,i}})\}_{u \in [1..n]}
\end{aligned}
\end{equation}

\item[3)] Both $P_1$ and $P_2$ publish all the commitments to all data providers.

\item[4)] $P_1$ opens $com(E_{1,u})$, $com(O_{2,u})$, and $P_2$ opens $com(E_{2,u})$, $com(O_{1,u})$, to each data provider $u$.

\item[5)] Each data provider $u$ decodes $O_{1,u}$ with $E_{1,u}$ to get $y_u^{(1)}$, decodes $O_{2,u}$ with $E_{2,u}$ to get $y_u^{(2)}$, and checks if $y_u^{(1)}=y_u^{(2)}$. If the check fails, $u$ proves this by publishing the openings in 4), namely $O_{1,u}$, $E_{1,u}$, $O_{2,u}$, $E_{2,u}$ and their corresponding random parameters, to other data providers.

\item[6)] If there is no fail proof, each data provider $u$ accepts its output $y_u$; otherwise, $u$ discards its output.
\end{itemize}

Note that in this phase, the data providers decode their outputs by themselves, without $P_1$ and $P_2$ learning anything about the plain outputs.

\subsection{Security Analysis}
In this section, we analyze the security of our framework against malicious adversaries. As we know, a malicious adversary may use any efficient attack strategies and may deviate from the protocol specification arbitrarily. The formal definition of security against malicious attackers refers to \cite{goldreich2009foundations}.

In our paper, a malicious adversary refer to the situations that one of parties ${P_1}$ and ${P_2}$ may be malicious, or a subset of data providers may be malicious. For instance, party ${P_1}$ or ${P_2}$ may tamper with the logic of the garbled circuits when it acts as a garbler, or the data providers may submit inconsistent inputs to both parties.

\begin{theorem}\label{the:malicious}
If either ${P_1}$ or ${P_2}$ is maliciously corrupted, or a subset of data providers are maliciously corrupted, our framework is secure against malicious adversaries.
\end{theorem}

We only sketch the proof. Since the roles of $P_1$ and $P_2$ are exactly symmetric, we analyze the security in two cases as follows.

\textbf{Malicious \bm{$P_1$}}: In this case, all data providers and $P_2$ are honest and follow the protocol. We analyze the security by distinguishing the three phases.

In the \emph{input consistency check} phase, all commitments are correctly constructed, and the adversary can only see all commitments, input encodings of the first garbled circuit $GC_1$ and input labels of the second garbled circuit $GC_2$. The adversary can learn nothing about the inputs. If it wants to deviate the protocol, all it can do is to falsify a check failure and abort the protocol. However, any check failure can be verified publicly by honest data providers and $P_2$, so the cheating of $P_1$ would be caught.

In the \emph{garbled circuit computation} phase, since the standard garbled circuit techniques, used in our work, are secure against malicious evaluators, we need only to analyze the security when $P_1$ runs $GC_1$ as a garbler. Because in the computation of $GC_1$, $P_1$ learns nothing about the inputs and outputs, the only attack is to generate an incorrect garbled circuit. However, this attack informs the adversary nothing about the inputs, since the outputs cannot be observed, and would be caught by the data providers in the output verification.

In the \emph{output verification} phase, the adversary only sees output encodings of $GC_1$ and output labels of $GC_2$, and learns nothing about the inputs and outputs. The adversary can attack by modifying the output labels, which would be caught by the data providers in the output verification.




The above analysis shows that the security holds in this case. The security of the case in which ${P_2}$ is malicious can be analyzed similarly.

\textbf{Malicious data providers.} In this case, both $P_1$ and $P_2$ are honest, and a subset of data providers (including the case of all data providers) are corrupted. We analyze the security by distinguishing both input consistency check and output verification phases, since data providers are not involved in the garbled circuit computation phase.

In the \emph{input consistency check}, we begin by proving the lemma below.

\begin{lemma}
If one input wire of a data provider $u$ has inconsistent labels for the two garbled circuits, the protocol aborts with probability of at least $1 - {2^{ - s + 1}}$.
\end{lemma}

\begin{proof}\renewcommand{\qedsymbol}{} 
The commitment construction check ensures that the labels in check sets must be all consistent. Otherwise, the check will fail. The label consistency check ensures that the labels in the input commitment sets of evaluation sets must be all consistent or all inconsistent (cf. Fig.~\ref{fig:check}). Otherwise, the check will also fail. Thus, the only successful cheating is the case that the labels in check sets are all consistent, and the labels in input commitment sets are all inconsistent, which happens with the probability $2^{-s}$. Furthermore, since there must be at least an evaluation set, and the case of $\rho=\left \langle 00\cdots 0 \right \rangle$ should be eliminated, the probability of cheating becomes $2^{-s+1}$.

The proof is completed. $\Box$
\end{proof}

If $s$ becomes large, the probability of cheating can be sufficiently small. The security against the malicious data providers can be adaptively achieved in term of $s$ value. Furthermore, any dishonest data provider will be identified exactly by a proof.

In the \emph{output verification} phase, both $P_1$ and $P_2$ commit to all the output encodings and labels of the two garbled circuits correctly, and open the commitments to each data provider honestly. The malicious data providers cannot change the output of the computation, since all the encodings and labels are correctly committed, and the output can be publicly verified if necessary.

From the discussion of the above two cases, we conclude that our protocol is secure against malicious adversaries.

\section{Application and Experiments}\label{sec:experiments}
In this section, we implement our secure computation framework instantiated with a cloud resource auction mechanism, and conduct extensive experiments to evaluate the performance in term of computation and communication overheads.

\subsection{Cloud Resource Auction}\label{sec:auction}

\subsubsection{Auction Mechanism}
We consider a truthful cloud resource auction mechanism from paper \cite{wang2012cloud}. In such an auction, a cloud provider provides $m$ types of virtual machine (VM) instances (with various configurations in CPU, memory, storage, etc. ). Each type $i$ has ${k_i}$ VM instances, and its weight is denoted by ${\omega _i}$, which is used to distinguish between different VM types. There are $n$ bidders requesting VM instances. Each bidder $j\left( {1 \le j \le n} \right)$ requests $k_j^i$ VM instances of type $i$, and the corresponding per-instance bid is $b_j^i$.
The auction proceeds as follows.

(1) The cloud provider computes the {\itshape average price per weighted instance} of each bidder $j$ :
\[ \overline{b}_j = \frac{{\mathop \sum \nolimits_{i = 1}^m k_j^ib_j^i}}{{\sqrt {\mathop \sum \nolimits_{i = 1}^m k_j^i{\omega _i}} }} \ , \ {\rm for} \ j \in [1..n] \]
and ranks the bidders in descending order of $\overline{b}_j$.

(2) After sorting, VM instances are allocated to the bidders by a greedy algorithm. If the VM instances requested by the bidder $j$ meet the following conditions:
\[\forall i \in \left\{ {1,2, \cdots ,m} \right\}\ ,\ \mathop \sum \nolimits_{t = 1}^{j - 1} k_t^i{x_t} + k_j^i \le {k_i},\]
the request is satisfied, and ${x_j} = 1$. Otherwise, ${x_j} = 0$.

(3) The cloud provider applies the following pricing rules to get a critical payment:
\begin{itemize}
  \item A bidder whose request is denied pays 0;
  \item A bidder having no critical bidder pays 0;
  \item A bidder who receives VM instances requested pays the critical value.
\end{itemize}
where the critical bidder $j_c$ of a bidder $j$ is the first bidder following $j$ that has failed but would win without $j$, and the critical value $v_c^j$ of $j$ is the minimum value $j$ should bid in order to win. It can be shown that
$$
v^j_c = \overline{b}_{j_c} \sqrt {\mathop \sum \nolimits_{i = 1}^m k_j^i{\omega _i}}
$$

Thus, the cloud provider charges each winning bidder $j$ ($x_j = 1$) the amount $v^j_c$.



\subsubsection{Data-oblivious Auction Algorithm}
To secure the cloud resource auction by the technique of garbled circuits, the auction mechanism must be first converted into a data-oblivious algorithm, whose execution path does not depend on the input values. An example of data-oblivious algorithm is illustrated in Fig.~\ref{fig:dataoblivious}. The data-oblivious algorithm can then be easily converted into a Boolean circuit. Instead of designing the data-oblivious algorithm ourself, we simply adopt one in paper \cite{chen2016privacy}. The details of the specific algorithm refer to Algorithm 2 in \cite{chen2016privacy}.
\begin{figure}[ht]
\centering
\includegraphics[width=0.49\textwidth]{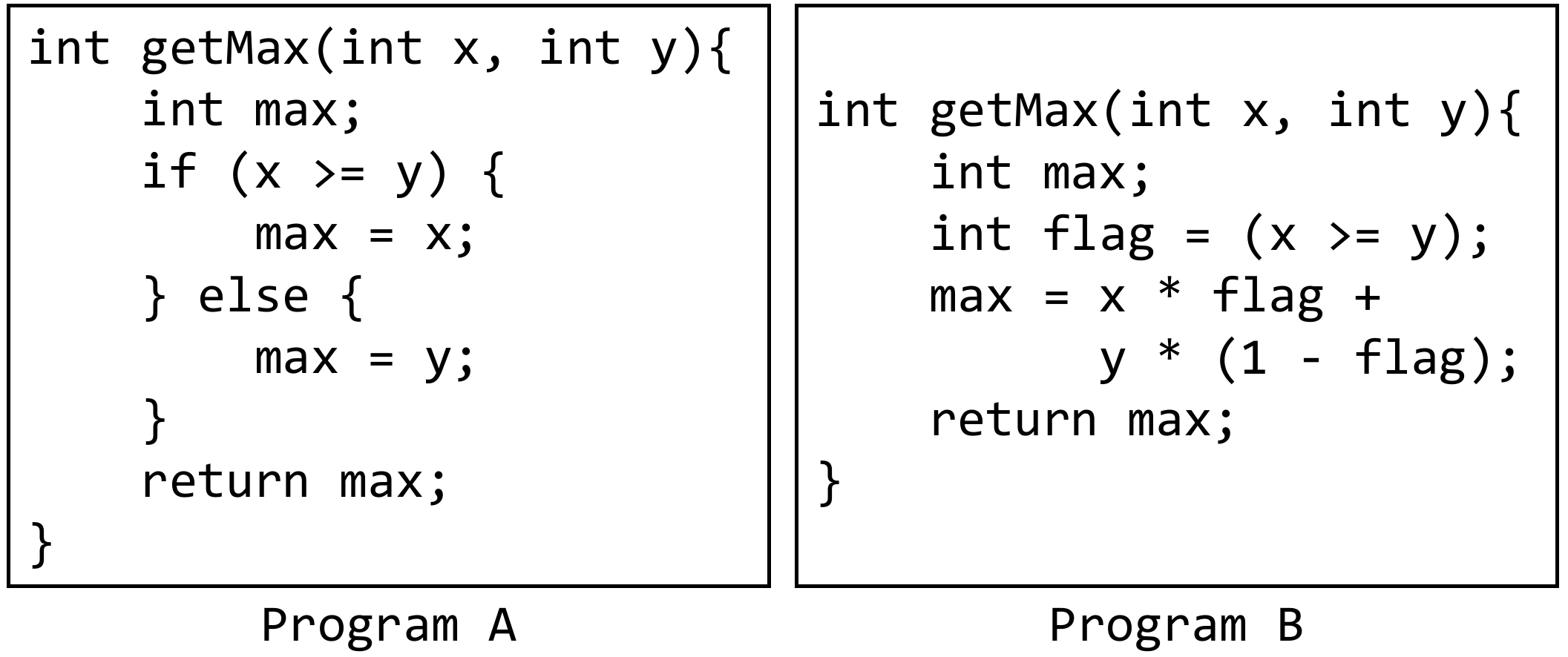}
\caption{\small An illustration of data-oblivious algorithms. Both algorithms find the maximum value of two numbers $x$ and $y$. The execution path of {\itshape Program A} depends on the input values $x$ and $y$, while that of {\itshape Program B} is the same regardless of the input values $x$ and $y$.}
\label{fig:dataoblivious}
\end{figure}

\subsection{Experimental Settings}
We instantiate our framework with the cloud resource auction mechanism in Section~\ref{sec:auction}. In the implementation, two non-colluding servers, which provide cryptographical services, are first determined to play the roles of computation parties $P_1$ and $P_2$. In practice, these two non-colluding servers may be from two security companies with well-deserved reputations, respectively. Then, the bidders play the roles of data providers, who submit their bids as inputs to $P_1$ and $P_2$, and receive their respective outputs from the two parties. Particularly, we let the cloud provider also play the role of a data provider, with the exception that the cloud provider submits no input, while receives all outputs.

We implement our framework on FastGC \cite{huang2011faster}, which is a Java-based framework for garbled circuit computation. We use SHA-256 hash function for computing hash values or doing commitment. In our simulation experiments, the number of virtual machine (VM) types is set $m = 6$ by default, and the number of instances for each VM type is set to a parameter $k$. The number of pairs of commitment sets for each input wire is set to $s = 10$, and all values in the auction are represented by $16$ bits. Each bidder $j\left( {1 \le j \le n} \right)$ requests a random number of VM instances ranging in $[0,3]$ with a per-instance bid ranging in $[0,100]$ for each VM type. We simulate our protocol on a 64-bit Windows 10 Desktop with Intel(R) Core(TM) i5 CPU @3.40GHz and 16GB of memory. In our experiments, we focus on the following performance metrics:
\begin{itemize}
  \item {\itshape Running time:} The total time spent by $P_1$, $P_2$ and the bidders during the execution of an auction.
  \item {\itshape Communication overhead:} The total volume of messages sent by $P_1$, $P_2$, and the bidders during the execution of an auction.
\end{itemize}

\subsection{Experimental Results}

We evaluate the performance of the protocol by varying the number of VM types $m$, the number of bidders $n$, and the instance count per VM type $k$.

We first set $m=6$, and evaluate the performance by varying $n$ and $k$ values. Fig.~\ref{fig:n100n150n200} shows the experimental results as $k$ increases from 50 to 300, for $n = 100$, $150$ and $200$, respectively. we can observe that in each case the total running time and communication overhead nearly remain the same, no matter how $k$ varies. Conversely, if $k$ is fixed, and $n$ increases, the running time and communication overhead will increase accordingly. This indicates that the performance of the protocol is hardly affected by the instance count per VM type $k$, but is significantly affected by number of bidders $n$.

Second, we set $n = 150$, and evaluate the performance by varying $m$ and $k$ values. Fig.~\ref{fig:m6m12m18n150} traces experimental results as $k$ increases from 50 to 250, for $m=6$, $12$, and $18$, respectively. We can see that the running time and communication overhead change little as $k$ increases when both $m$ and $n$ values are fixed. Conversely, when $k$ is fixed, and $m$ increases from 6 to 18, the running time and communication overhead will increase approximately linearly. This shows that the performance of the protocol is hardly affected by the instance count per VM type $k$, but is significantly affected by number of VM types $m$.

Third, we set $m=6$, and evaluate the performance by varying $n$ and $k$ values. In Fig.~\ref{fig:k75k150}, we set $k = 75$ and $150$ respectively, and track the experimental results as $n$ increases from 50 to 300. We can observe that, when $k$ is fixed, in other words, the resources are limited, the running time and communication overhead increase super-linearly as $n$ increases. When fixing $n$, the differences in running time and communication overhead is small for $k =75$ and $k = 150$. We can observe the same as that from Fig.~\ref{fig:n100n150n200}.

Finally, we set $k=100$, and evaluate the performance by varying $n$ and $m$ values. Fig.~\ref{fig:m6m12m18k100} illustrates the comparison of the running time and communication overhead as $n$ increases when adopting different $m$ values. We fix $k$ to 100 and set $m$ to 6, 12 and 18, respectively. We can observe that the running time and communication overhead increase super-linearly as $n$ increases from 50 to 250. Moreover, when fixing $n$ value, the running time and communication overhead will increase significantly as $m$ increases. This demonstrates that the performance of the protocol is significantly affected by both the instance count per VM type $k$ and number of VM types $m$.

From these experimental analysis above, we can conclude that the performance is affected approximately linearly by $m$, super-linearly by $n$, and hardly by $k$. Furthermore, when the number of bidders $n$ takes values of several hundreds, the running time is tens of minutes, and the communication overhead is thousands of MB, on our PC platform. This is acceptable in practice for some small-scale applications at least.

\begin{figure}[htbp]
  \centering
  \subfigure[]{\includegraphics[width=0.222\textwidth]{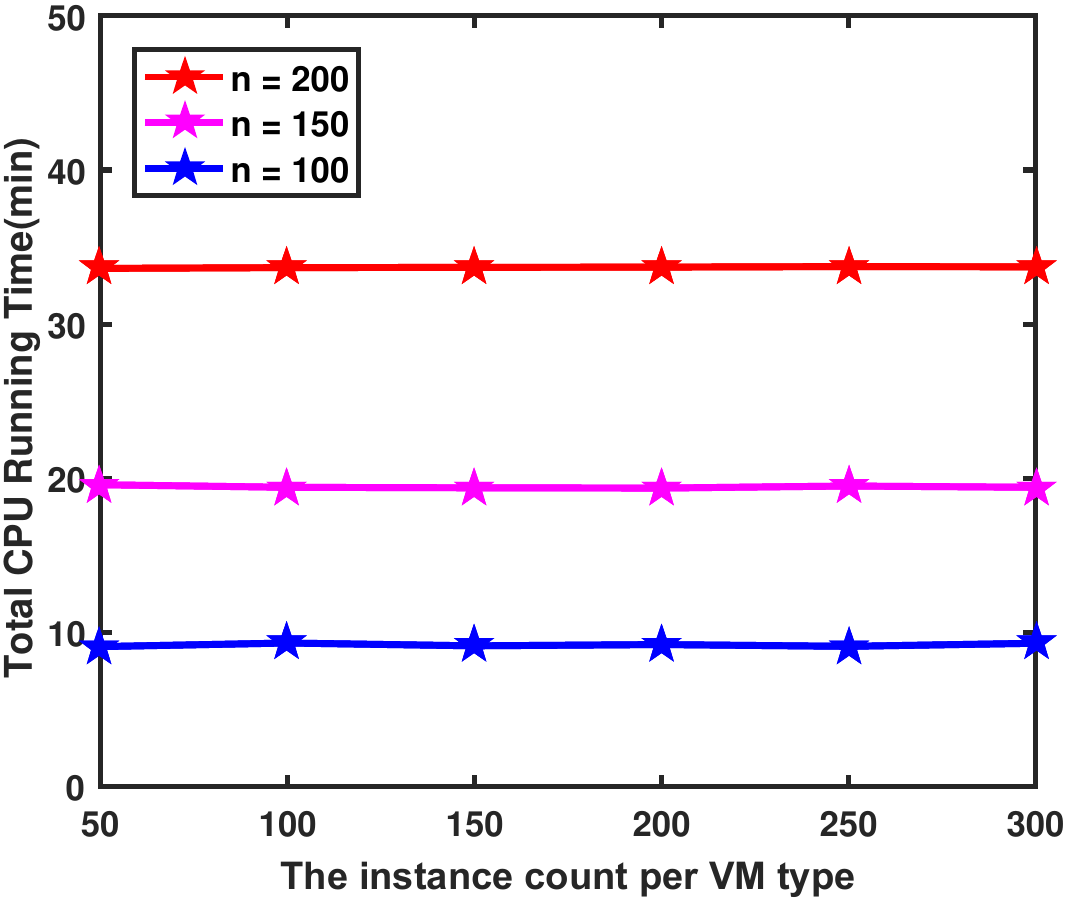}}
  \quad
  \subfigure[]{\includegraphics[width=0.228\textwidth]{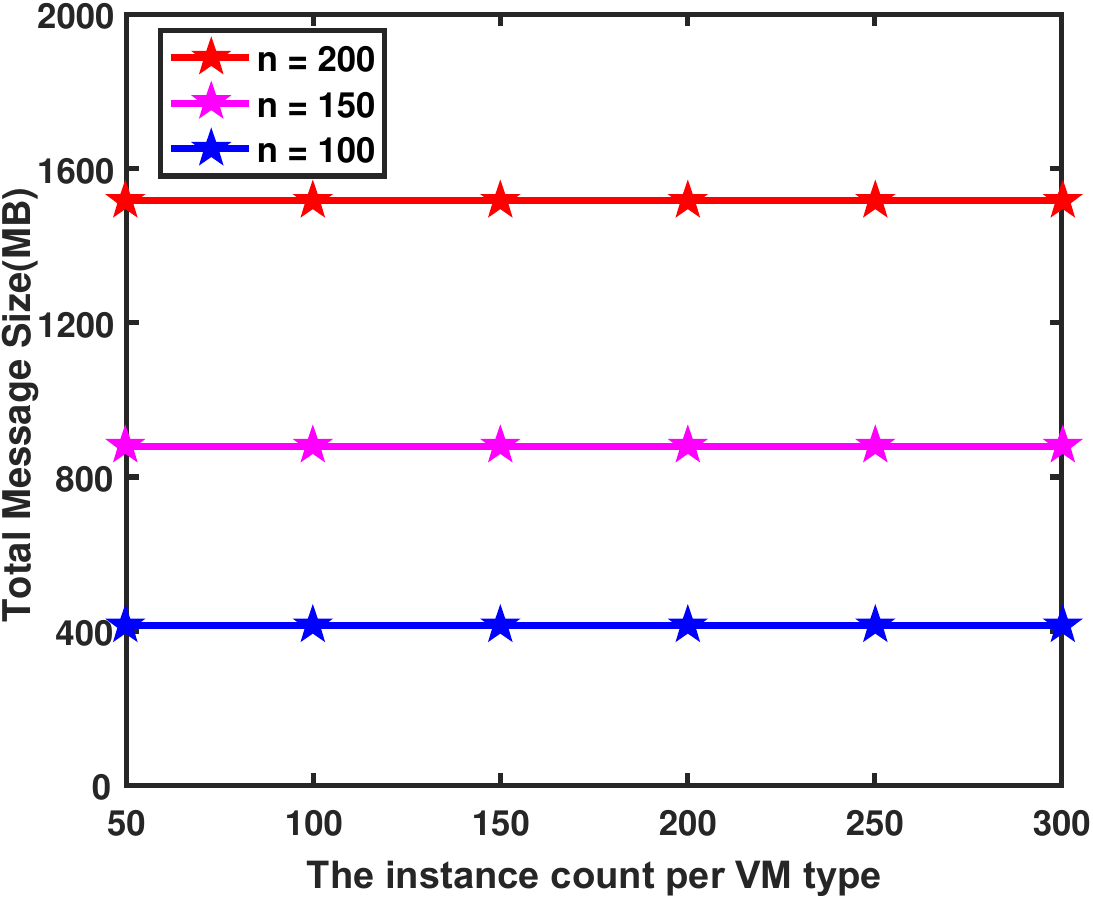}}
  \caption{Performance comparison as the instance count per VM type $k$ varies for three different values of the number of bidders $n$.}\label{fig:n100n150n200}
\end{figure}
\begin{figure}[htbp]
  \centering
  \subfigure[]{\includegraphics[width=0.222\textwidth]{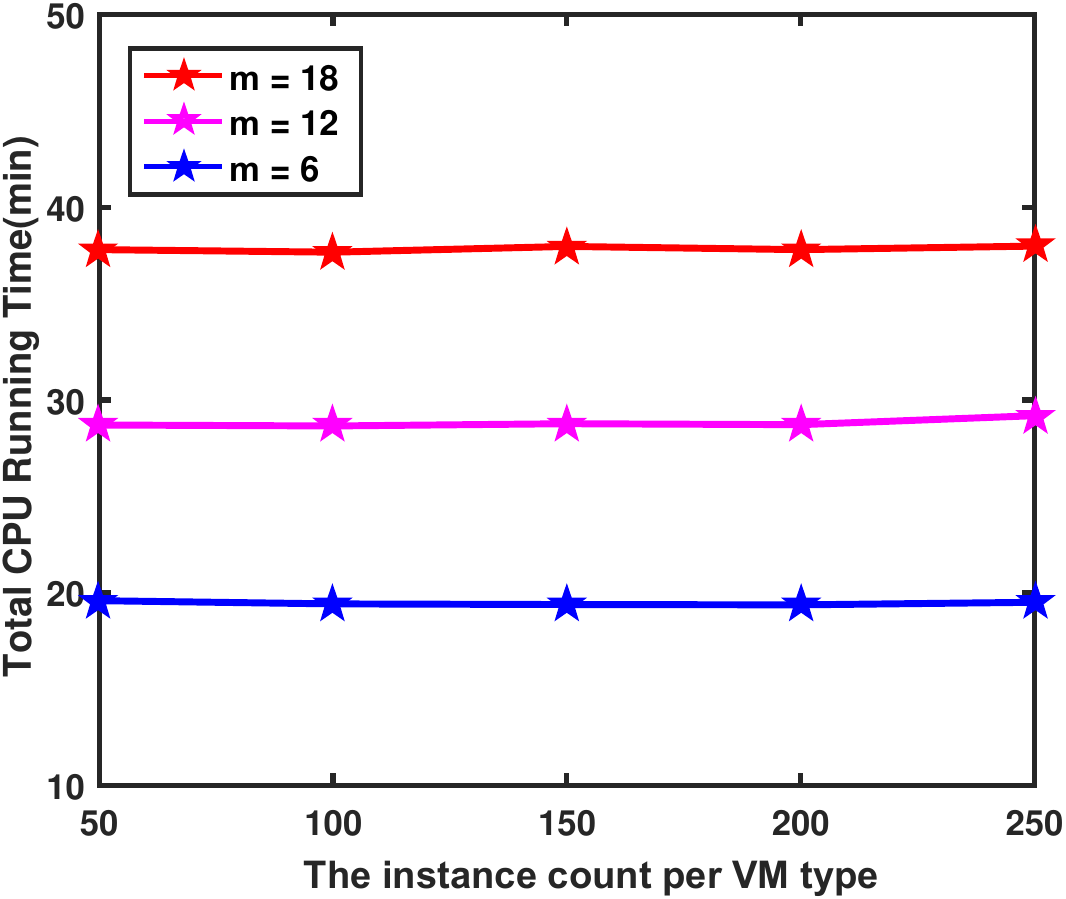}}
  \quad
  \subfigure[]{\includegraphics[width=0.228\textwidth]{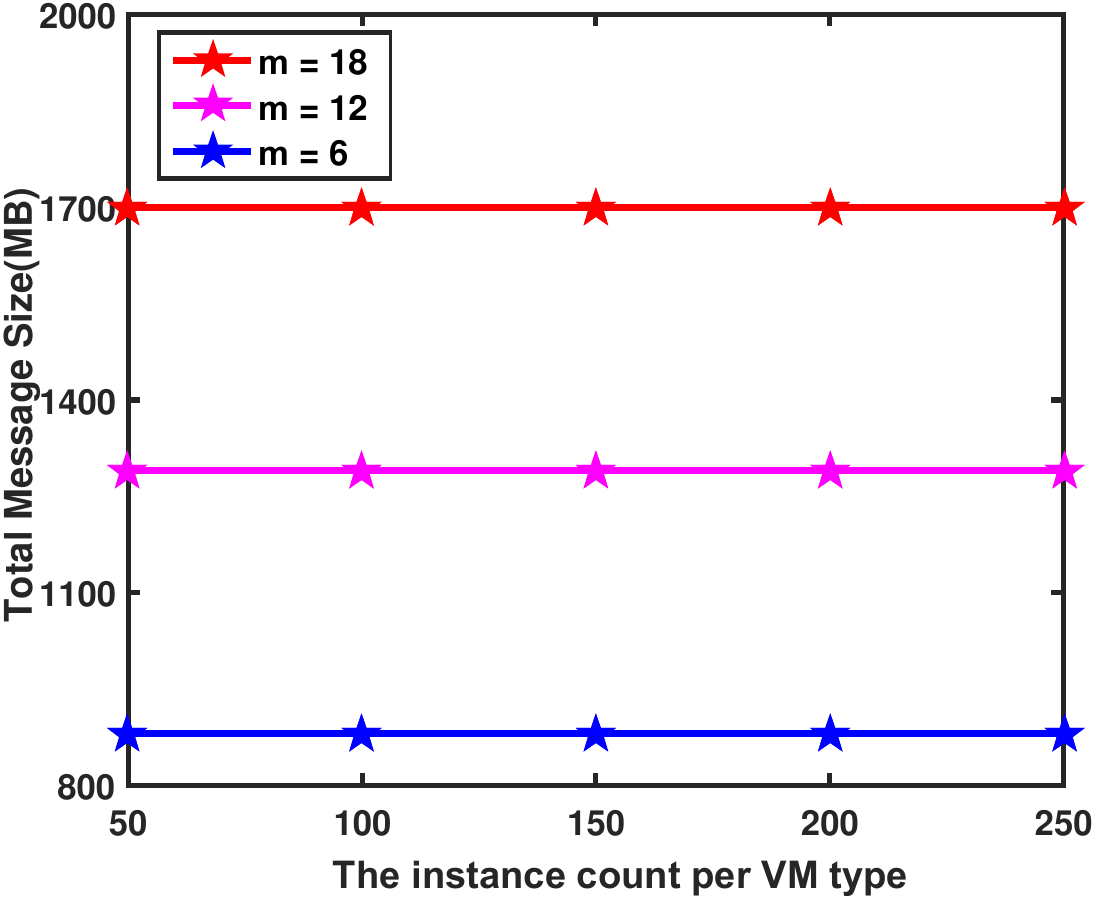}}
  \caption{Performance comparison as the instance count per VM type $k$ varies for three different values of the number of VM types $m$.}\label{fig:m6m12m18n150}
\end{figure}
\begin{figure}[htbp]
  \centering
  \subfigure[]{\includegraphics[width=0.222\textwidth]{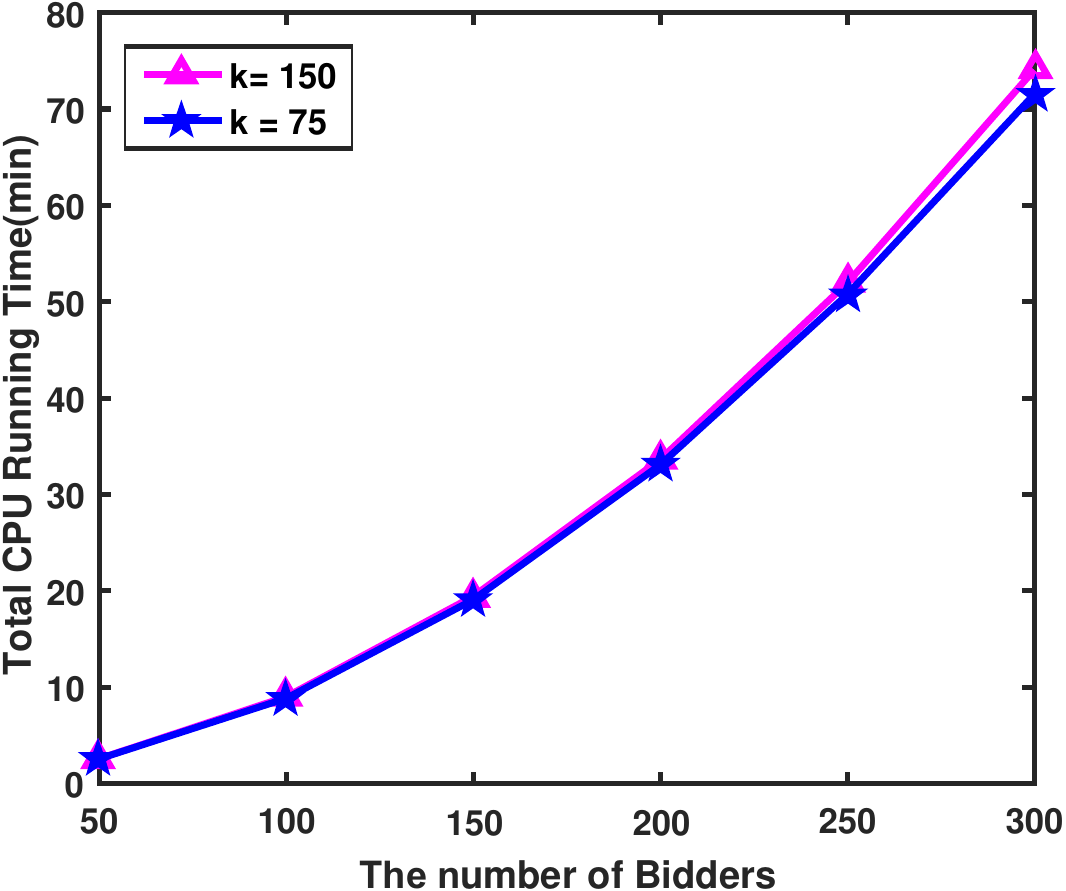}}
  \quad
  \subfigure[]{\includegraphics[width=0.228\textwidth]{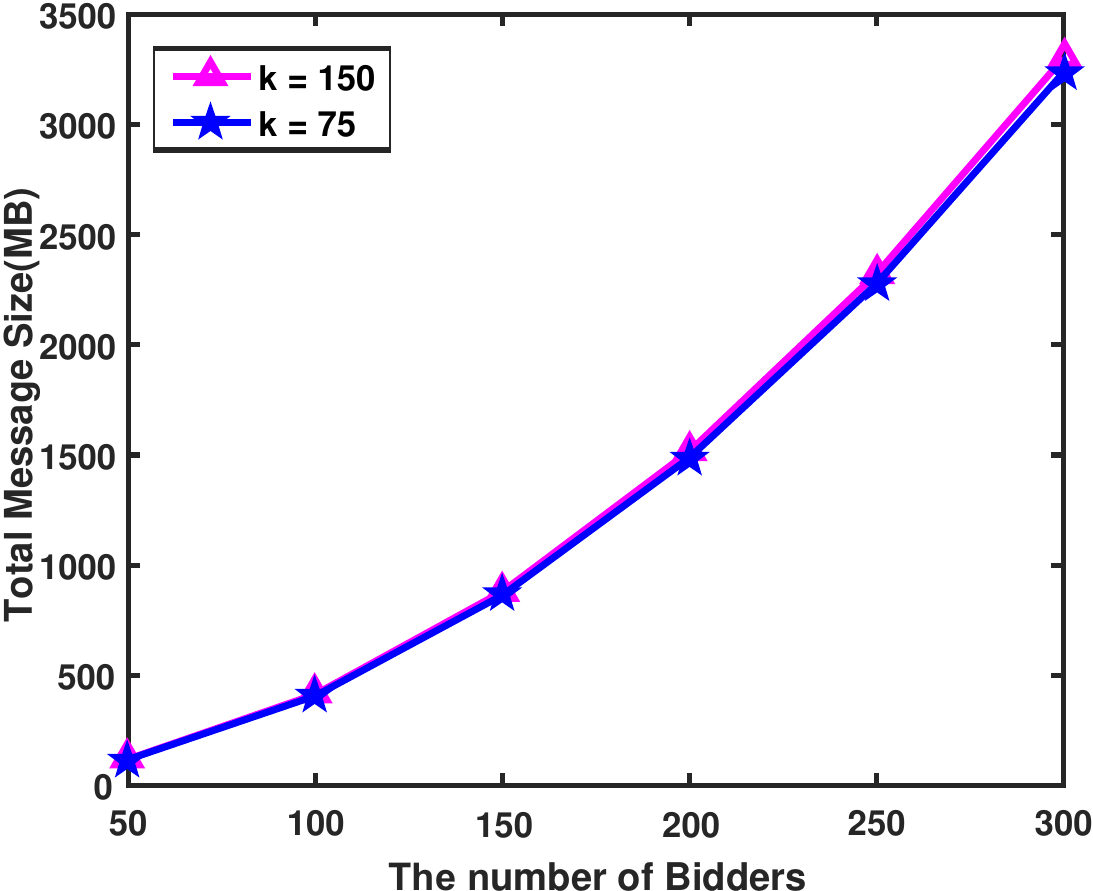}}
  \caption{Performance comparison as the number of bidders $n$ varies for two different values of the instance count per VM type $k$.}\label{fig:k75k150}
\end{figure}

\begin{figure}[htbp]
  \centering
  \subfigure[]{\includegraphics[width=0.222\textwidth]{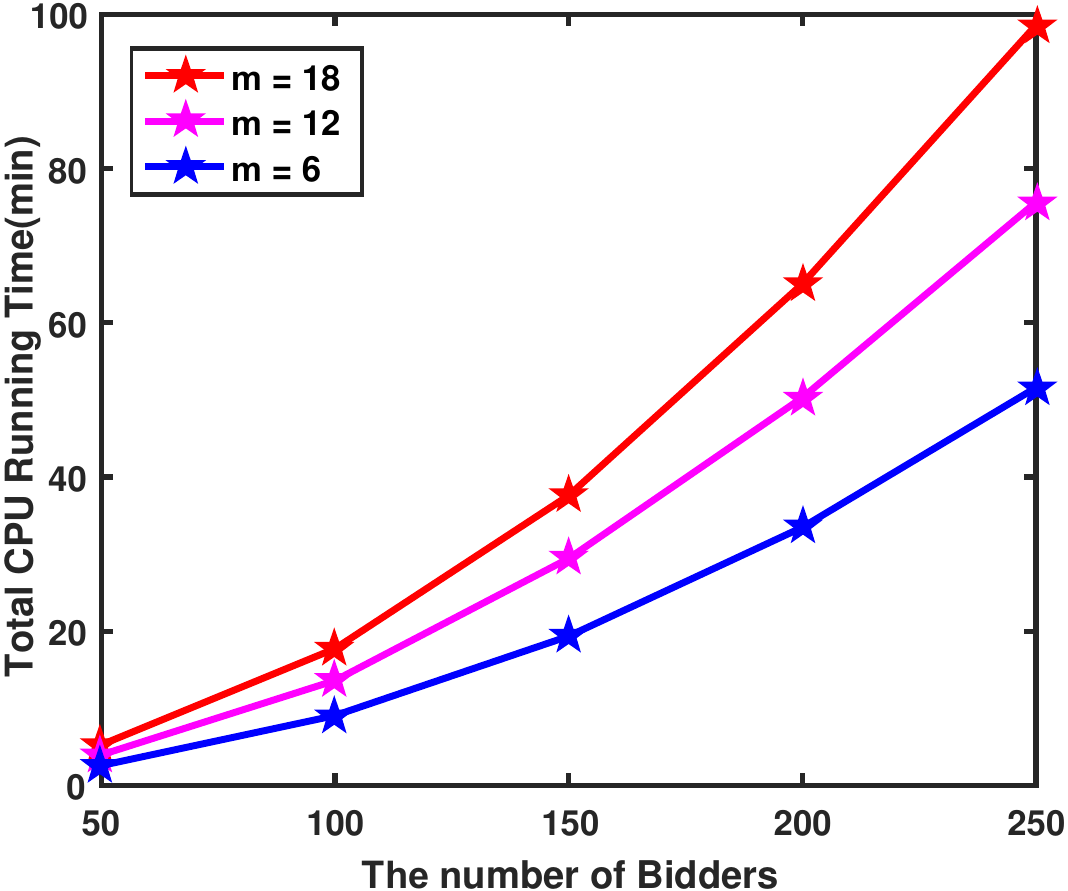}}
  \quad
  \subfigure[]{\includegraphics[width=0.228\textwidth]{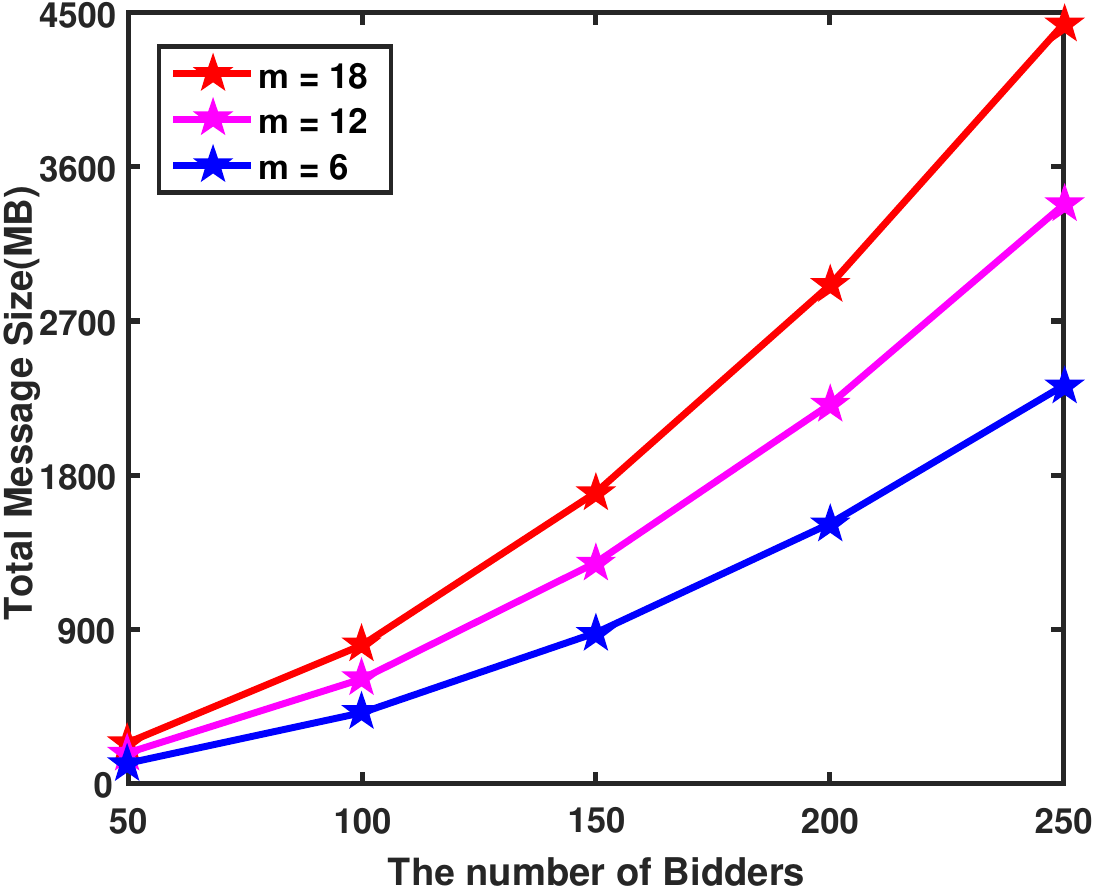}}
  \caption{Performance comparison as the number of bidders $n$ varies for three different values of the the number of VM types $m$.}\label{fig:m6m12m18k100}
\end{figure}

\section{Conclusions}\label{sec:conclusion}
In this paper, we have presented a general secure computation framework for multiple data providers. Making use of the specific context, our framework is able to achieve security against malicious adversaries by running only two independent garbled circuits, with the roles of computation parties swapped. To defend against malicious data providers, we have designed an input consistency check mechanism using the cut-and-choose paradigm. This mechanism enables that all data providers submit the same inputs for the two garbled circuits, and is able to identify accurately the data providers who cheat by submitting inconsistent inputs with a cheating proof. We have also designed an output verification mechanism, such that the correctness of outputs can be verified publicly by all the data providers, without both computation parties knowing anything about the output. We theoretically prove the security of the entire framework. Furthermore, we have implemented our framework in a cloud resource auction scenario, and extensive experimental results show that the performance of our framework is acceptable in practical applications.

\bibliographystyle{ACM-Reference-Format}
\bibliography{ref}

%
%
%
%
%
%
%
%

\newpage

\end{document}